%% file: paper.tex
\documentclass{llncs}

\usepackage[utf8]{inputenc}

\usepackage[british]{babel}
\usepackage{microtype}

\usepackage{url}

\usepackage{multibib}
\newcites{supp}{Appendix References}

\usepackage[dvipsnames]{xcolor}

\usepackage{amsmath,amssymb,amsfonts}
\usepackage{stmaryrd}
\usepackage{algorithm}
\usepackage{algpseudocode}
\usepackage{graphicx}
\usepackage{textcomp}
\usepackage{nccmath}

\usepackage{blindtext}

\usepackage[labelfont=bf,font=footnotesize]{caption}
\usepackage{subcaption}

\usepackage{comment}
\usepackage[hidelinks]{hyperref}

\newcommand\thefont{\expandafter\string\the\font}

\usepackage[makeroom]{cancel}

\usepackage{mathbbol}

\usepackage{mathrsfs}

\usepackage{xcolor}

\usepackage[shortlabels]{enumitem}

\usepackage[backgroundcolor=white]{todonotes}

\usepackage{soul}

\usepackage{yfonts}

\usepackage[title]{appendix}

\usepackage{listings}

\usepackage{makecell}

\usepackage{calc}

\usepackage{tikz}
\usetikzlibrary{automata, arrows.meta, positioning, calc}

\usepackage{xstring}
\usepackage[shortcuts]{extdash}

\usepackage{pifont}
\newcommand{\cmark}{\color{ForestGreen}{\ding{51}}}%
\newcommand{\xmark}{\color{BrickRed}{\ding{55}}}%

\makeatletter
\providecommand{\bigsqcap}{%
  \mathop{%
    \mathpalette\@updown\bigsqcup
  }%
}
\newcommand*{\@updown}[2]{%
  \rotatebox[origin=c]{180}{$\m@th#1#2$}%
}
\makeatother

\DeclareCaptionFormat{algor}{%
  \hrulefill\par\offinterlineskip\vskip1pt%
    \textbf{#1#2}#3\offinterlineskip\hrulefill}
\DeclareCaptionStyle{algori}{singlelinecheck=off,format=algor,labelsep=space}
\newcommand*{\defeq}{\stackrel{\text{def}}{=}}
\def\generatepks{\ensuremath\Gamma{}}
\newcommand*{\MakePKS}[1]{\generatepks(#1)}

\newcommand*{\x}{\bot}

\newcommand*{\CommentSR}[1]{\todo{\color{red} #1}}
\newcommand*{\CommentJO}[1]{\todo{\color{blue} #1}}

\newcommand*{\ResponseSR}[1]{{\color{red} -- #1}}
\newcommand*{\ResponseJO}[1]{{\color{blue} -- #1}}

\newcommand*{\mck}[0]{\textbf{machine-check}}

\newdimen\numheight
\newdimen\numwidth
\newcommand{\resizedx}{\resizebox{\numwidth}{\numheight}{X}}

\newcommand{\tvbv}[1]{\StrSubstitute{#1}{X}{\resizedx}}

\newenvironment{kgraph}
{
\begin{tikzpicture}  [node distance = 1.35cm, on grid, auto, initial distance=0.3cm,
    every initial by arrow/.style = {-stealth},
    every state/.style={inner sep=0.02cm, minimum size=0.65cm},
    every edge/.style = {draw,-stealth},
    every loop/.style = {-stealth},
    ]
}{
\end{tikzpicture}}

\def\kgreen{LimeGreen!80!black}
\def\korange{Peach!70!white}
\def\kgrey{Gray!70!white}

\newcommand{\krstate}[3][\null]{
\setlength{\numheight}{\heightof{1}}
\setlength{\numwidth}{\widthof{1}}
  \IfBeginWith{#2}{0}{  
    \node (#2) [state, fill=\kgreen, #1] {#3};
  }{
    \IfBeginWith{#2}{1}{
        \node (#2) [state, fill=\korange, #1] {#3};
    }{
        \node (#2) [state, fill=\kgrey, #1] {#3};
    }
}
}

\newcommand{\kstate}[2][\null]{\krstate[#1]{#2}{\tvbv{#2}}}

\spnewtheorem*{sketch}{Proof sketch}{\itshape}{\rmfamily}

\allowdisplaybreaks

\begin{document}

\title{
Input-based\\Three-Valued Abstraction Refinement
}

\author{%
Jan Onderka\inst{1,2,3}\orcidID{0000-0003-2069-8584} \and
Stefan Ratschan\inst{2}\orcidID{0000-0003-1710-1513}
}%
\authorrunning{
J. Onderka and S. Ratschan
}
\institute{
Faculty of Engineering, University of Freiburg, Freiburg, Germany \\
\email{onderka@cs.uni-freiburg.de} \\ \and
Institute of Computer Science, The Czech Academy of Sciences, \\ Prague, Czech Republic \\
\email{stefan.ratschan@cs.cas.cz} \\ \and
Faculty of Information Technology, Czech Technical University in Prague \\ Prague, Czech Republic
}

\maketitle

\begin{abstract}

Unlike Counterexample-Guided Abstraction Refinement (CEGAR), Three-Valued Abstraction Refinement (TVAR) is able to verify all properties of the \textmu-calculus. We present a novel algorithmic framework for TVAR that employs a simulator-like approach to build and refine the abstract state space with input-based splitting. This leads to a state space formalism that is much simpler than in previous TVAR frameworks, which use modal transitions. We implemented the framework in our open-source tool \mck{} and verified properties of machine-code systems for the AVR architecture, showing the ability to verify systems and \textmu-calculus properties not verifiable by naïve model checking or CEGAR, respectively. This is the first practical use of TVAR for machine-code verification.
\end{abstract}


\input{sources/01_introduction}
\input{sources/02_previous}
\input{sources/03_formalism}
\input{sources/04_choices}

\input{sources/05_evaluation}

\input{sources/06_conclusion}

\bibliographystyle{splncs04}
\bibliography{IEEEabrv,ref}

\begin{subappendices}
\renewcommand{\thesection}{\Alph{section}}

\newpage
\input{sources/a_proofs}


\end{subappendices}

\end{document}

%% file: sources/01_introduction.tex
\section{Introduction}
\label{sec:introduction}

Abstraction-refinement methodologies are ubiquitous in formal verification, the foremost being Counterexample-guided Abstraction Refinement (CEGAR) \cite{clarke2002,clarke2000}. Unfortunately, CEGAR does not support the whole propositional \textmu-calculus, and indeed not even Computation Tree Logic (CTL). This leaves a large class of potentially crucial non-linear-time properties unverifiable. Three-valued Abstraction Refinement (TVAR) is able to verify full \textmu-calculus. However, previous TVAR frameworks refined abstract states, necessitating state space formalisms based on modal transitions. Frameworks based on simple formalisms~\cite{godefroid2001,shoham2003} are not monotone: previously provable properties may no longer be provable after refinement. Intricate monotone formalisms \cite{shoham2004,gurfinkel2006,wei2011} were devised, their specialised semantics complicating the use of standard model-checking algorithms.

Another problem common to model-checking with abstraction is that the abstract state space cannot be built using a system simulator, which is possible for naïve explicit-state model checking \cite{schlich2006}. This is a problem especially with complex systems such as processors executing machine code. 

\textbf{Contribution.} Combining the ideas of TVAR and system simulators, we present a novel TVAR framework based on simulator-like generation of the abstract state space and performing refinements by splitting abstract inputs. Our framework does not use modal transitions, leading to simpler and more intuitive reasoning compared to the previous frameworks. Furthermore, it allows simple building and refinement of the abstract state space based on abstract simulators. We prove that the introduced framework is sound, monotone, and complete for \textmu-calculus properties and existential abstraction domains, provided simple requirements are met. Arbitrary digital systems formalised as automata can be verified, and the choices of abstraction domains and refinement strategy can be tailored to the specific use-case. Unlike the previous abstract-simulator approaches, our framework can be used for the full \textmu-calculus.

We implemented an instance of our framework in our formal verification tool \mck{}\footnote{Free and open-source, homepage at \url{https://machine-check.org/}. We discuss version 0.6.1, published at \url{https://crates.io/crates/machine-check/0.6.1}.}, where the systems are described as simulable finite-state machines in a subset of the Rust language, translated to abstract and refinement analogues used for generating and refining the state space \cite{onderka2024meco}. The ability to refine removes the need for hints such as where to use abstraction \cite{gueckel2012}.

We designed \mck{} especially for machine-code verification, and were able to verify non-toy programs for the AVR architecture and find a bug in a simplified version of a real-life program using a non-linear-time property not verifiable by CEGAR. To our knowledge, this is the first use of TVAR and verification of \textmu-calculus properties for machine-code systems.

%% file: sources/02_previous.tex
\section{Previous Work}
\label{sec:previous}

In this section, we list previous relevant work on TVAR in roughly chronological order, with additional information available in summarising papers \cite{dams2018,godefroid2014}. After that, we discuss relevant work on abstract simulator approaches to model checking. 

\begin{example}
\label{ex:system}
Consider the finite-state machine in Figure~\ref{fig:system}, representing e.g.\@ a controller of aircraft landing gear retraction: if the most significant bit (\emph{msb}) of the state is 0, the landing gear is extended; if 1, it is retracted. The system is required to follow a single-bit input from the gear lever, with some slack for responses. There is a critical bug, occurring if the aircraft loses power in flight when the landing gear is retracted and the controller restarts in the state 000 after the power is regained: since the landing gear lever is set to retraction, the controller goes through the states 011 and 111, causing a \textbf{total loss of capability} to extend the gear again unless the controller is turned off and on again.

The bug is not just dangerous, but also sneaky, as it does not occur in normal aircraft operation. To protect ourselves against it, we can verify the property ``from every reachable system state, it should be possible to reach a state where the landing gear is extended'' holds in the system. This is formalised by a \emph{recovery property} \textbf{AG}[\textbf{EF}[$\lnot$\textit{msb}]] in CTL, not possible to check using CEGAR.

Since the system is buggy, the property should be disproved\footnote{In our terminology, \textit{proving} the property determines it holds in the system. \textit{Disproving} it determines it does not. \textit{Verification} aims to either prove or disprove it.}. We will instead reason about proving a dual property, \textbf{EF}[\textbf{AG}[\textit{msb}]]. In Figure \ref{fig:system}, \textit{msb} definitely holds in the state 101, and since 101 just loops on itself, \textbf{AG}[\textit{msb}] holds in it. As 101 is reachable from 000, \textbf{EF}[\textbf{AG}[$\textit{msb}$]] holds, and the bug is found. While such reasoning is easy for simple systems, in real life, the controller may have billions of possible states, requiring us to abstract some information away\footnote{The example is directly inspired by a bug we found, discussed in Section \ref{sec:evaluation}.}.
\end{example}

\textbf{Partial Kripke Structures (PKS).} In verification on \textit{partial state spaces}~\cite{bruns99}, some information is disregarded to produce a smaller state space. PKS enrich standard Kripke structures (KS) by allowing unknown state labellings.

\begin{figure}[t]
\centering
\begin{kgraph}
    \kstate[initial]{000}
    \kstate[right = 1.6cm of 000]{001}
    \kstate[below = 1.4cm of 001]{011}
    \kstate[right = 1.6cm of 001]{010}
    \kstate[right = 1.6cm of 011]{111}
    \kstate[right = 1.6cm of 010]{110}
    \kstate[right = 1.6cm of 111]{100}
    \kstate[right = 1.6cm of 100]{101}
    
    \path [-stealth]
    (000) edge node {0} (001)
    (000) edge node {1} (011)
    (001) edge node {0,1} (010)
    (011) edge node {0,1} (111)
    (010) edge [bend left] node {1}   (110)
    (010) edge [loop above]  node {0}()
    (111) edge node {0} (100)
    (111) edge [bend left=37] node {1} (101)
    (110) edge [bend left] node {0}   (010)
    (110) edge [loop above]  node {1}()
    (100) edge node {0,1} (101)
    (101) edge [loop above]  node {0,1}();
\end{kgraph}
\caption{Example system expressed as a finite-state machine. The states where $\lnot$\textit{msb} holds are \textcolor{\kgreen}{drawn green} while the states where \textit{msb} holds are \textcolor{\korange}{drawn orange}.}
\label{fig:system}
\end{figure}

\begin{definition}
\label{def:pks}
A partial Kripke structure (PKS) is a tuple $(S, S_0, R, L)$ over a set of atomic propositions $\mathbb{A}$ with the elements
\begin{itemize}
    \item $S$ (the \emph{set of states}),
    \item $S_0 \subseteq S$ (the \emph{set of initial states}),
    \item $R \subseteq S \times S$ (the \emph{transition relation}),
    \item $L: S \times \mathbb{A} \rightarrow \{0, 1, \x\}$ indicating for each atomic proposition whether it holds, does not hold, or its truth value is unknown (the \emph{labelling function}).
\end{itemize}
A Kripke Structure (KS) is a PKS with $L$ restricted to $S \times \mathbb{A} \rightarrow \{0, 1\}$.
\end{definition}

\begin{example}
While Figure \ref{fig:system} shows a finite-state machine, it can be converted to a Kripke structure by discarding the inputs, with $S = \{000, 001, \dots, 111\}$, $S_0 = \{ 000 \}$, and $R$ given by the transitions in Figure \ref{fig:system}. $L$ labels $msb$ in states $\{000, 001, 010, 011\}$ as 0, and in $\{100, 101, 110, 111\}$ as 1.

For proving \textbf{EF}[\textbf{AG}[$\textit{msb}$]], such a KS is unnecessarily detailed. Using PKS, we could e.g. combine 010 and 110 into a single \textit{abstract} state where it is unknown what the value of the most significant bit is, and the labelling of $\textit{msb}$ is $\bot$.
\end{example}

\textbf{Existential abstraction.} In TVAR, \textit{existential abstraction} \cite{clarke1994} is used, where the abstract states in set $\hat{S}$ are related to the original concrete states in $S$ by a function $\gamma: \hat{S} \rightarrow 2^S$, the abstract state $\hat{s} \in \hat{S}$ representing some (not fixed) concrete state in $s \in \gamma(\hat{s})$ in each system execution instant. This is a generalisation of domains usable for Abstract Interpretation \cite{cousot1977,cousot1979}, also allowing non-lattice domains such as wrap-around intervals \cite{gange2015}.

\begin{example}
In the examples, we will use the three-valued bit-vector domain, where each element is a tuple of three-valued bits, each with value `0' (definitely 0), `1' (definitely 1), or `X' (possibly 0, possibly 1). Except for figures, we write three-valued bit-vectors in quotes, e.g. $\gamma(\text{``0X1''}) = \{ 001, 011 \}$. The bits can also refer to a predicate rather than a specific value. For example, `1' could mean that $v > 5$ holds, `0' that its negation holds, and `X' that we do not know.
\end{example}

\subsection{Previous TVAR Frameworks}
\label{subsec:previous_tvar}
Building on the work of Bruns \& Godefroid \cite{bruns99,bruns2000,bruns2001}, Godefroid et al. \cite{godefroid2001} introduced TVAR by refining the abstract state set, using a state space formalism based on modal transitions. Early TVAR approaches \cite{godefroid2001,godefroid2002,shoham2003,grumberg2005,grumberg2007} were based on Kripke Modal Transition Structures (KMTS) and did not guarantee previously provable properties stay provable after refinement, i.e.\@ were not monotone.
\begin{definition}
\label{def:kmts}
A Kripke Modal Transition Structure (KMTS) is a five-tuple $(S, S_0, R^{\textrm{may}}, R^{\textrm{must}}, L)$ where $S, S_0,$ and $L$ follow Definition \ref{def:pks}, and
\begin{itemize}
    \item $R^{\textrm{may}} \subseteq S \times S$ is the set of transitions which may be present,
    \item $R^{\textrm{must}} \subseteq R^{\textrm{may}}$ is the set of transitions which are definitely present.
\end{itemize}
\end{definition}

Intuitively, KMTS allow for transitions with unknown presence ($R^{\textrm{may}} \setminus R^{\textrm{must}}$). PKS can be trivially converted to KMTS by setting $R^{\textrm{may}} = R^{\textrm{must}} = R$. While it is possible to convert a KMTS to an equally expressive PKS by moving the transition presence into the states \cite{godefroid2003}, this requires the set of states to be modified.

\textbf{Monotone frameworks.} Godefroid et al. recognised non-monotonicity as a problem and suggested keeping previous states when refining \cite[p.~3-4]{godefroid2001}. However, Shoham \& Grumberg showed the approach was not sufficient, since in certain cases, it prevents verification of additional properties after refinement. As a remedy, they introduced another monotone TVAR framework using Generalized KMTS for CTL \cite{shoham2004}, later extended to \textmu-calculus \cite{shoham2008}. Gurfinkel \& Chechik introduced a framework for verification of CTL properties on Boolean programs using Mixed Transition Systems \cite{gurfinkel2006}, later extended to lattice-based domains~\cite{gurfinkel2008}\footnote{An instance of the framework is implemented in the tool Yasm \cite{gurfinkel2006yasm}, available online at the time of writing \cite{yasm}. However, it does not support common programming language elements such as bitwise-operation statements (for example, $y = x \; \& \; 1$) nor full \textmu-calculus, which our tool can handle without problems.}. Wei et al. introduced a TVAR framework using Reduced Inductive Semantics for \textmu-calculus under which the results of model-checking GKMTS and MixTS are equivalent~\cite{wei2011}.

\begin{definition}
A Generalized KMTS (GKMTS) is a tuple $(S, S_0, R^{\textrm{may}}, R^{\textrm{must}}, L)$ where $S, S_0, R^{\textrm{may}},$ and $L$ follow Definition \ref{def:kmts} and $R^{\textrm{must}}: S \times 2^S$ is the set of \emph{hyper-transitions}, where $\forall (a, B) \in R^{\textrm{must}} \; . \; \forall b \in B \; . \; (a,b) \in R^{\textrm{may}}$.
\end{definition}

\begin{definition}
A Mixed Transition System (MixTS) is a tuple $(S, S_0, R^{\textrm{may}},$ $R^{\textrm{must}}, L)$ where $S, S_0, R^{\textrm{may}},$ and $L$ follow Definition \ref{def:kmts} and $R^{\textrm{must}} \subseteq S \times S$.
\end{definition}

\begin{figure}[t]
\centering
\begin{subfigure}[b]{0.25\textwidth}
\begin{kgraph}
    \kstate[initial]{XXX}
    \path [-stealth]
    (XXX) edge [loop right]  node {}();
\end{kgraph}
\caption{Only state ``XXX''.}
\label{fig:state_based_initial}
\end{subfigure}
\begin{subfigure}[b]{0.3\textwidth}
\begin{kgraph}
    \kstate[initial]{0XX}
    \kstate[right = of 0XX] {1XX}
    \path [-stealth]
    (0XX) edge [dashed, loop above]  node {}()
    (0XX) edge [dashed, bend left] node {}(1XX)
    (1XX) edge [dashed, bend left] node {}(0XX)
    (1XX) edge [dashed, loop above]  node {}();
\end{kgraph}
\caption{KMTS after refining (a).}
\label{fig:state_based_kmts}
\end{subfigure}
\begin{subfigure}{0.3\textwidth}
\begin{kgraph}
    \kstate[initial]{0XX}
    \kstate[right = of 0XX] {1XX}
    \coordinate (M1) at ($(0XX)!0.5!(1XX)$);
    \kstate[below = 0.8cm of M1] {XXX}
    \path [-stealth]
    (0XX) edge [dashed, loop below]  node {}()
    (0XX) edge [dashed, bend left] node {}(1XX)
    (1XX) edge [dashed, bend left] node {}(0XX)
    (1XX) edge [dashed, loop below]  node {}()
    (0XX) edge (XXX)
    (1XX) edge (XXX)
    (XXX) edge [dashed, loop left] (XXX)
    (XXX) edge [loop right] (XXX)
    ;
\end{kgraph}
\caption{MixTS after refining (a).}
\label{fig:state_based_mixts}
\end{subfigure}

\begin{subfigure}{0.25\textwidth}
\begin{kgraph}
    \kstate[initial]{0XX}
    \kstate[right = 1cm of 0XX] {1XX}
    \coordinate (M1) at ($(0XX)!0.3!(1XX)+(0,0.4)$);
    \coordinate (M2) at ($(1XX)!0.3!(0XX)-(0,0.4)$);
    \path [-,thick]
    (0XX) edge [-] (M1)
    (1XX) edge [-] (M2);
    \path [-stealth]
    (M1) edge [thick, out=45, in=90, looseness=4]  node {}(0XX)
    (M1) edge [thick, out=45, in=120, looseness=1] node {}(1XX)
    (M2) edge [thick, out=-135, in=-60, looseness=1] node {}(0XX)
    (M2) edge [thick, out=-135, in=-90, looseness=4]  node {}(1XX);
\end{kgraph}
\caption{GKMTS after refining (a).}
\label{fig:state_based_gkmts}
\end{subfigure}
\begin{subfigure}{0.25\textwidth}
\begin{kgraph}
    \kstate[initial]{0XX}
    \kstate[right = 1cm of 0XX] {11X}
    \kstate[below = 1.2cm of 11X] {10X}
    \coordinate (M1) at ($(0XX)!0.3!(11X)+(0,0.4)$);
    \coordinate (M3) at ($(0XX)!0.5!(10X)$);
    \coordinate (M4) at ($(M3)!0.4!(11X)$);
    
    \path [-,thick]
    (0XX) edge [-] (M1)
    (11X) edge [-] [thick] (M4)
    (10X) edge [thick, loop above] ()
    ;
    \path [-stealth]
    (M1) edge [thick, out=45, in=90, looseness=4]  node {}(0XX)
    (M1) edge [thick, out=45, in=120, looseness=1] node {}(11X)
    (M4) edge [thick, out=-135, in=-60, looseness=1.5] (0XX)
    (M4) edge [thick, out=-135, in=135, looseness=1.5] (10X)
    ;
\end{kgraph}
\caption{GKMTS after refining ``1XX'' in (d).}
\label{fig:state_based_cont}
\end{subfigure}
\begin{subfigure}{0.45\textwidth}
\centering
\begin{kgraph}
    \kstate[initial]{000}
    \kstate[right = 1.2cm of 000]{001}
    \kstate[below = 1.1cm of 001]{011}
    \kstate[right = 1.2cm of 001]{010}
    \kstate[right = 1.2cm of 011]{111}
    \kstate[right = 1.2cm of 010]{110}
    \kstate[right = 1.2cm of 111]{10X}
    
    \path [-stealth]
    (000) edge [thick] (001)
    (000) edge [thick] (011)
    (001) edge [thick] (010)
    (011) edge [thick] (111)
    (010) edge [thick, bend left]   (110)
    (010) edge [thick, loop above] ()
    (111) edge [thick] (10X)
    (110) edge [thick, bend left] (010)
    (110) edge [thick, loop above] ()
    (10X) edge [thick, loop above] ();
\end{kgraph}
\caption{GKMTS after refining ``0XX'' in (e) by splitting to ``000'', ``001'', ``010'', and ``011''.}
\label{fig:state_based_final}
\end{subfigure}

\caption{State-based refinement with hyper-transitions based on Generalised KMTS. Implied may-transitions, present in all sub-figures except for (c), are not drawn. The states where it is unknown whether \textit{msb} or $\lnot$\textit{msb} holds are \textcolor{\kgrey}{drawn grey}.}
\label{fig:state_based}
\vspace{-0.4em}
\end{figure}

\begin{example}
\label{ex:kmts}
    We will prove \textbf{EF}[\textbf{AG}[\textit{msb}]] using the previous state-based TVAR frameworks over the system from Figure \ref{fig:system}, starting with abstract state set \{``XXX''\}. Clearly, we both \textbf{may} and \textbf{must} transition from ``XXX'' to ``XXX'', visualised in Figure \ref{fig:state_based_initial}. Since $\textit{msb}$ is unknown in ``XXX'', the model-checking result is unknown and we refine.
    
    Suppose we decide to refine by splitting the abstract state set to \{``0XX'', ``1XX''\}, starting in ``0XX''. From ``0XX'', we \textbf{may} transition either to ``0XX'' (e.g. by 000 $\rightarrow$ 001 or 010 $\rightarrow$ 010) or ``1XX'' (e.g. by 010 $\rightarrow$ 110), but cannot conclude that e.g. a transition from ``0XX'' to itself \textbf{must} exist: $011 \in \gamma(\text{``0XX''})$ only transitions to $111 \not\in \gamma(\text{``0XX''})$.

    PKS cannot be used as they cannot describe unknown-presence transitions. KMTS allow this, producing Figure \ref{fig:state_based_kmts}. However, it is not possible to prove e.g.\@ \textbf{EX}[\textbf{true}], which was possible in Figure \ref{fig:state_based_initial}, i.e.\@ the refinement is not monotone. Using MixTS, we retain ``XXX'' and the must-transitions to it, producing a \emph{forced choice} in Figure \ref{fig:state_based_mixts}. Using GKMTS, we obtain Figure \ref{fig:state_based_gkmts} instead. In both, it is possible to prove \textbf{EX}[\textbf{true}], but not \textbf{EF}[\textbf{AG}[\textit{msb}]]. Refining further using GKMTS, we obtain Figure \ref{fig:state_based_cont}, where it is still not possible to prove \textbf{EF}[\textbf{AG}[\textit{msb}]]: the hyper-transitions do not imply that the path (``0XX'',``11X'',``10X'') corresponds to a concrete path. The property is only proven after additional refinement to Figure \ref{fig:state_based_final}. While the GKMTS in Figure \ref{fig:state_based_final} trivially corresponds to a KMTS or a PKS, fewer refinements and final states may be needed in general when using GKMTS or MixTS as they may guide the refinement better due to monotonicity.

\end{example}

\textbf{Model checking.} \textmu-calculus properties can be model-checked on PKS and KMTS by a simple conversion to two KS, applying standard model-checking algorithms, and combining the results \cite{bruns2000,godefroid2002}. Similar conversions are also possible for multi-valued logics \cite{konikowska2002,gurfinkel2003}. It is also possible to model-check directly without conversion. A multi-valued model-checker was previously used for the MixTS approach \cite{gurfinkel2006}. Game-based model-checking was used for full \textmu-calculus \cite{grumberg2005,grumberg2007}.

\textbf{Discussion of previous TVAR frameworks.} It was recognised early on that using KMTS with non-monotone refinement is problematic \cite{godefroid2001,shoham2003}. GKTMS seem more susceptible to exponential explosion as MixTS can make use of abstract domains. However, specialised algorithms must be used for MixTS to obtain GKMTS-equivalent results~\cite{wei2011}. A drawback of all mentioned approaches is their conceptual complexity which, in our opinion, is the main reason for the dearth of available TVAR tools and test sets, compared to CEGAR. This has also made the analysis of these methods difficult, as illustrated by the subtle differences in definitions of expressiveness identified by Gazda \& Willemse~\cite{gazda2012}.

\subsection{Abstract Simulator Approaches}
\label{subsec:simulation}
Naïve explicit-state model checking generates a next state from every state and input combination, with the ability to use a system simulator to perform each step. This simulator can be written in an imperative language such as C. Unfortunately, for machine-code systems where state sizes are in kilobytes even for simple microcontrollers and a single port read can produce e.g. $2^8$ next states, this results in exponential explosion for all but the simplest toy programs. We will now discuss approaches where the state space is abstract, but still built by generating the next states by an \emph{abstract simulator}.

\textbf{Trajectory Evaluation.} Bryant used three-valued logic simulators for formal verification of hardware circuits~\cite{bryant1991a,bryant1991c}. He showed linear-time properties (expressed by specification machines or circuit assertions) can be proven using a set of three-valued input sequences that together cover all concrete inputs \cite[p.~320]{bryant1991a} by generating permissible state sequences (\emph{trajectories}). 
\emph{Symbolic trajectory evaluation (STE)} is an extension that allows parametrisation of the introduced \emph{trajectory formulas} \cite{bryant1991b}. However, the STE formalism drops the distinction between inputs and states. We refer to Melham~\cite{Melham2018} for a discussion of STE and extensions. Notably, Generalized STE~\cite{yang2001} allows verification of properties corresponding to linear-time \textmu-calculus~\cite{dam1992}. While manual refinement was originally needed, automatic refinement was proposed for both STE \cite{tzoref2006} and GSTE \cite{chen2007}.

\textbf{Delayed Nondeterminism.} Noll \& Schlich \cite{noll2008} verified machine-code programs by model-checking an abstract state space generated by a simulation-based approach. Each input bit was read as `X' and split to `0' and `1' only when it was decided to in a subsequent step (e.g. if it was an argument of a branch instruction).  This allowed e.g. splitting only one bit of a read 8-bit port if the other bits were masked out by a constant first, allowing soundly proving (but not disproving) properties in path-universal logics such as LTL and ACTL.

\begin{figure}[t]
\centering
\begingroup
\renewcommand{\arraystretch}{0.3}
\begin{subfigure}{0.7\textwidth}
\centering
\begin{kgraph}
    \krstate[initial]{0_000}{\tvbv{000}}
    \krstate[right = of 0_000]{1_0X1}{\tvbv{0X1}}
    \krstate[right = of 1_0X1]{X_X1X}{\tvbv{X1X}}
    \krstate[right = of X_X1X]{X_XXX}{\tvbv{XXX}}

    \path [-stealth]
    (0_000) edge node {} (1_0X1)
    (1_0X1) edge node {} (X_X1X)
    (X_X1X) edge node {} (X_XXX)
    (X_XXX) edge [loop right]  node {}();
\end{kgraph}
\caption{Abstract state space corresponding to the system in Figure \ref{fig:system}, computed by simulation with unconstrained inputs (`X')}
\label{fig:simulation_normal}
\end{subfigure}

\begin{subfigure}{0.8\textwidth}
\centering
\begin{kgraph}
    
    \krstate[initial, initial text={\textbf{Case 1.} system}]{0_000}{\tvbv{000}}
    \krstate[right = of 0_000]{1_001}{\tvbv{001}}
    \krstate[right = of 1_001]{0_010}{\tvbv{010}}
    \krstate[right = of 0_010]{X_X10}{\tvbv{X10}}
    
    \krstate[circle split, initial, initial text={specification}, below = 0.8cm of 0_000]{X_0spec0}{\tvbv{00} \nodepart{lower} X}
    \krstate[circle split,right = of X_0spec0]{X_0spec1}{\tvbv{01} \nodepart{lower} X}
    \krstate[circle split,right = of X_0spec1]{0_0spec2}{\tvbv{10} \nodepart{lower} 1}
    \krstate[circle split,right = of 0_0spec2]{X_0spec3}{\tvbv{11} \nodepart{lower} X}
    
    \krstate[initial, initial text={\textbf{Case 2.} system}, below = 1.1cm of X_0spec0]{0_000'}{\tvbv{000}}
    \krstate[right = of 0_000']{1_011}{\tvbv{011}}
    \krstate[right = of 1_011]{0_111}{\tvbv{111}}
    \krstate[right = of 0_111]{X_10X}{\tvbv{10X}}
    
    \krstate[circle split,initial, initial text={specification}, below = 0.8cm of 0_000']{X_1spec0}{\tvbv{00} \nodepart{lower} X}
    \krstate[circle split,right = of X_1spec0]{X_1spec1}{\tvbv{01} \nodepart{lower} X}
    \krstate[circle split,right = of X_1spec1]{0_1spec2}{\tvbv{10} \nodepart{lower} 1}
    \krstate[circle split,right = of 0_1spec2]{X_1spec3}{\tvbv{11} \nodepart{lower} X}
    
    \path [-stealth]
    (0_000) edge node {0} (1_001)
    (1_001) edge node {X} (0_010)
    (0_010) edge node {X} (X_X10)

    (X_0spec0) edge node {0} (X_0spec1)
    (X_0spec1) edge node {X} (0_0spec2)
    (0_0spec2) edge node {X} (X_0spec3)
    
    (0_000') edge node {1} (1_011)
    (1_011) edge node {X} (0_111)
    (0_111) edge  node {X} (X_10X)
    
    (X_1spec0) edge node {1} (X_1spec1)
    (X_1spec1) edge node {X} (0_1spec2)
    (0_1spec2) edge node {X} (X_1spec3)
    ;
\end{kgraph}
\caption{Trajectory evaluation: The verification is split into cases, ensuring the abstract input sequences together cover all possible concrete input sequences. Unlike Bryant, we use initial states for consistency with other approaches.}
\label{fig:simulation_bryant}
\end{subfigure}

\begin{subfigure}{0.8\textwidth}
\centering
\begin{kgraph}
    \krstate[initial]{0_000}{\tvbv{000}}
    \krstate[right = of 0_000]{1_0X1}{\tvbv{0X1}}
    \krstate[right = of 1_0X1]{0_010}{\tvbv{010}}
    \krstate[below = 1.0cm of 0_010]{0_111}{\tvbv{111}}
    \krstate[right = of 0_010]{X_X10}{\tvbv{X10}}
    \krstate[right = of 0_111]{X_10X}{\tvbv{10X}}
    
    \path [-stealth]
    (0_000) edge node {} (1_0X1)
    (1_0X1) edge [dashed, yshift=-0.05cm] node {d.001} (0_010)
    (1_0X1) edge [dashed, swap] node {d.011} (0_111)
    (0_010) edge node {} (X_X10)
    (0_111) edge  node {} (X_10X)
    (X_X10) edge [loop right]  node {}()
    (X_10X) edge [loop right]  node {}();
\end{kgraph}
\caption{Delayed nondeterminism augmented with must-transitions: the `X' in state ``0X1'' is split to `0' and `1' before computing the successor}
\label{fig:simulation_delayed}
\end{subfigure}
\endgroup

\caption{Simulation-based approaches proving \textbf{A}[\textbf{X}[\textbf{X}[\textit{msb} $\Leftrightarrow$ \textit{lsb}]]]. (a) can be considered PKS or KMTS, and (c) KTMS. (b) contains four trajectories. Specially in this figure, system states are drawn \textcolor{\kgreen}{green} if $msb \Leftrightarrow lsb$ holds, \textcolor{\korange}{orange} if it does not, and \textcolor{\kgrey}{grey} if it is unknown. Specification states are coloured according to the output.}
\vspace{-0.3em}
\label{fig:simulation}
\end{figure}

\begin{example}
Due to the restrictions of the approaches, we will illustrate proving the property ``in two steps from the initial state, the most significant bit corresponds to the least significant bit'', i.e.\@ \textbf{A}[\textbf{X}[\textbf{X}[\textit{msb} $\Leftrightarrow$ \textit{lsb}]]]. Simulating without splitting, we produce Figure~\ref{fig:simulation_normal}, unable to prove the property.

To visualise Bryant's trajectory evaluation approach with explicitly considered inputs \cite{bryant1991a}, we encode the specification as a finite-state machine with two bits containing an initially-zero saturating counter. The system output function is \textit{msb} $\Leftrightarrow$ \textit{lsb}. The specification outputs `1' iff the counter is 10 and `X' otherwise. To prove the property, we split verification into two cases based on the value of the first input, and obtain simulated trajectories of both machines in Figure \ref{fig:simulation_bryant}. The property is proven as the trajectories are long enough (at least 3 for the given property) and the specification output always covers the system output.

To better understand Delayed Nondeterminism, we augment with must-transitions where possible. Splitting ``0X1'' from Figure~\ref{fig:simulation_normal}, we obtain Figure \ref{fig:simulation_delayed}, where ``010'' is obtained as a direct successor of ``001'', and ``111'' as a direct successor of ``011''. We cannot augment during the split as `X' might not generally correspond to a unique input, potentially e.g.\@ being copied before splitting.
\end{example}

%% file: sources/03_formalism.tex
\section{Input-based Three-valued Abstraction Refinement}
\label{sec:automata}

We propose a framework that eliminates the need for modal transitions in TVAR by combining ideas from the discussed approaches: using TVAR, build the abstract state space using an abstract simulator and split \textbf{inputs} instead of states.

\begin{figure}[t]
\centering
\begin{subfigure}{0.495\textwidth}
\centering
\begin{kgraph}
    \kstate[initial]{000}
    \kstate[right = 1.1cm of 000]{001}
    \kstate[below = 1.03cm of 001]{011}
    \kstate[right = 1.1cm of 001]{010}
    \kstate[right = 1.1cm of 011]{111}
    \kstate[right = 1.1cm of 010]{X10}
    \kstate[right = 1.1cm of 111]{10X}
    
    \path [-stealth]
    (000) edge node {0} (001)
    (000) edge [swap] node {1} (011)
    (001) edge node {X} (010)
    (011) edge node {X} (111)
    (010) edge node {X} (X10)
    (111) edge  node {X} (10X)
    (X10) edge [loop right]  node {X}()
    (10X) edge [loop right]  node {X}();
\end{kgraph}
\caption{Refining the input after ``000'' while keeping all others `X'.}
\label{fig:ours_refine}
\end{subfigure}
\begin{subfigure}{0.495\textwidth}
\centering
\begin{kgraph}
    \kstate[initial]{000}
    \kstate[right = 1.1cm of 000]{XXX}
    \kstate[below = 1.03cm of 001]{011}
    \kstate[right = 1.1cm of 011]{111}
    \kstate[right = 1.1cm of 111]{10X}
    
    \path [-stealth]
    (000) edge node {0} (XXX)
    (000) edge [swap] node {1} (011)
    (011) edge node {X} (111)
    (111) edge node {X} (10X)
    (XXX) edge [loop right]  node {X}()
    (10X) edge [loop right]  node {X}();
\end{kgraph}
\caption{A smaller reachable state space using a \textit{decayed} step function.}
\label{fig:ours_decay}
\end{subfigure}
\caption{Input-based Abstraction Refinement.}
\label{fig:ours}
\vspace{-0.75em}
\end{figure}

\begin{example}
We return to the original problem of proving \textbf{EF}[\textbf{AG}[$msb$]]. The simulation-based approach initially builds the abstract state space as shown in Figure \ref{fig:simulation_normal}. After that, we decide (using e.g. a heuristic, machine learning or human guidance) that the input after ``000'' should be split. We regenerate the abstract state space as shown in Figure \ref{fig:ours_refine}. We are immediately able to prove \textbf{EF}[\textbf{AG}[$msb$]] holds, meaning the system from Figure \ref{fig:system} contains a bug.

The part of the abstract state space in Figure \ref{fig:ours_refine} starting with ``001'' is unnecessarily large for proving the property, potentially causing exponential explosion problems. To prevent them, we also introduce a way to soundly and precisely regulate the outgoing states of transitions, allowing us to e.g. replace ``001'' by ``XXX'' as in Figure \ref{fig:ours_decay} when generating the abstract state space, \emph{decaying} to less information. Only one refinement was necessary compared to multiple in Example \ref{ex:kmts}, with the final state space in Figure \ref{fig:ours_decay} smaller than in Figure \ref{fig:state_based_final}. However, this depends\footnote{As with other frameworks, verification performance depends drastically on abstraction, refinement, and implementation choices, further discussed in Sections \ref{sec:choices} and \ref{sec:evaluation}.} on the correct choice to decay ``001'' and not ``011''.

\end{example}

The generated state spaces are PKS, which allows us to use previous work on PKS, KMTS, GKMTS, and MixTS, as PKS are trivially convertible to all. This notably includes model-checking using standard formalisms \cite{bruns2000,godefroid2002} and TVAR refinement guidance \cite{shoham2003,grumberg2005,grumberg2007}, with the caveat that we need to select an input instead of a state to refine. Unlike (G)STE \cite{bryant1991b,yang2001} and Delayed Nondeterminism~\cite{noll2008} which were limited to linear-time or path-universal properties, our approach can be used for the full \textmu-calculus. Verification can be fully automatic or manually guided, and it is also possible to precisely control the number of reachable abstract states and transitions: we can split inputs up to one by one, and decay any newly reachable states before refining the applied decay.

We will now give the framework formalism and simple requirements for its instances to be sound, monotone, and complete, proving that the requirements are sufficient in Appendix \ref{app:proofs}. We will then discuss how reasonable choices of refinements can be made in Section \ref{sec:choices}. Finally, we will evaluate an implementation of an instance of our framework in \mck{} in Section \ref{sec:evaluation}.

\subsection{Framework Formalism}

We assume that the original Kripke Structure has only one initial state\footnote{This is merely a formal choice. For multiple initial states, a dummy initial state can be introduced before them and the verified property $\phi$ converted to \textbf{AX}[$\phi$].}, i.e. $K=(S, \{s_0\}, R, L)$. We write the result of model-checking a property $\phi$ against $K$ as $\llbracket \phi \rrbracket(K)$, which returns 0 or 1. For a PKS~$\hat{K}$, $\llbracket \phi \rrbracket(\hat{K})$ returns 0, 1, or $\bot$. 

We consider the original (concrete) system to be an automaton and will also use automata for abstracting the system, introducing the formalism of \emph{generating automata} that can generate partial Kripke structures.

\begin{definition}
    A \emph{generating automaton} (GA) is a tuple $G=(S, s_0, I, q, f, L)$ with the elements
\begin{itemize}
    \item $S$ (the \emph{set of automaton states}),
    \item $s_0 \in S$ (the \emph{initial state}),
    \item $I$ (the \emph{set of all step inputs}),
    \item $q: S \rightarrow 2^I \setminus \{ \emptyset \}$ (the \emph{input qualification function}), 
    \item $f: S \times I \rightarrow S$ (the \emph{step function}),
    \item $L: S \times \mathbb{A} \rightarrow \{0,1,\x\}$ (the \emph{labelling function}).
\end{itemize}
\end{definition}

\begin{definition}
\label{def:generatepks}
For a generating automaton $(S, s_0, I, q, f, L)$, we define the PKS-generating function $\generatepks$ as
\begin{align}
    \generatepks((S, s_0, I, q, f, L)) \defeq & \; (S, \{s_0\}, R, L) \label{eq:generatepks} \\
    \text{where }  R = & \; \{ (s, f(s, i)) \; | \; s \in S, i \in q(s) \}.\label{eq:concrete_generation}
\end{align}
\end{definition}

\goodbreak

We call a generating automaton $G=(S, s_0, I, q, f, L)$ \emph{concrete} if the labelling function $L: S \times \mathbb{A} \rightarrow \{0,1\}$ (disallowing the value $\x$) and for all $s\in S$, $q(s)=I$.
A concrete GA corresponds to a Moore machine with the output of each state mapping each atomic proposition from $\mathbb{A}$ to either 0 or 1.

\begin{center}
\vspace{-1em}
  \captionof{algorithm}{Input-based Three-valued Abstraction Refinement Framework}\label{alg:framework}
{ \footnotesize
\begin{algorithmic}
\State \hspace{-1em} \textbf{Require}: a concrete generating automaton $(S, s_0, I, q, f, L)$, a \textmu-calculus property $\phi$
\State \hspace{-1em} \textbf{Ensure}: return $\llbracket\phi\rrbracket((S, s_0, I, q, f, L))$ \Comment{If requirements are fulfilled, see Corollary \ref{col:soundness}}
\State
    \State $(\hat{S}, \hat{s}_0, \hat{I}, \hat{q}, \hat{f}, \hat{L}) \gets \Call{Abstract}{S, s_0, I, q, f, L}$
    \While{$(r \gets \llbracket \phi \rrbracket(\Call{\generatepks}{(\hat{S}, \hat{s}_0, \hat{I}, \hat{q}, \hat{f}, \hat{L})}) = \bot$}
        \State $(\hat{q}, \hat{f}) \gets \Call{Refine}{\hat{S}, \hat{s}_0, \hat{I}, \hat{q}, \hat{f}, \hat{L}}$
    \EndWhile
    \State \Return $r$
\end{algorithmic}
\vspace{-0.75em}
\hrulefill
\vspace{-0.5em}
}
\end{center}

Algorithm \ref{alg:framework} describes our framework. Given a concrete GA, it abstracts it to an \emph{abstract} generating automaton $(\hat{S}, \hat{s}_0, \hat{I}, \hat{q}, \hat{f}, \hat{L})$, successively refining the input qualification function $\hat{q}$ and step function $\hat{f}$ until the result of model-checking is non-$\x$. $\hat{S}$ and $\hat{I}$ are related to $S$ and $I$ by a state concretization function\footnote{We forbid abstract elements with no concretizations as they do not represent any concrete element. Practically speaking, this does not disqualify abstract domains with such elements; we just require such elements are not produced by $\hat{s}_0$, $\hat{q}$, or $\hat{f}$.} $\gamma: \hat{S} \rightarrow 2^S \setminus \{\emptyset\}$ and an input concretization function $\zeta: \hat{I} \rightarrow 2^I \setminus \{\emptyset\}$.

Unlike state-based TVAR, the set of abstract states $\hat{S}$ does not change during refinement. The number of states to be considered is limited by the codomain of~$\hat{f}$, allowing structures such as Binary Decision Diagrams to be used. The abstract state space can be built quickly by forward simulation. For backward simulation, care must be taken to pair the states according to inputs.

\begin{example}
\label{ex:mck_ga}
In \mck, states and inputs are composed of bit-vector and bit-vector-array variables, formally represented by flattened $S = \{ 0, 1 \}^w$ and $I = \{ 0, 1 \}^y$ for finite state width $w$ and finite input width $y$. A dummy $s_0$ precedes the actual initial system states, $f$ is a function written in an imperative programming language, and $L$ computes relational operations on state variables.

For the abstract GA, we use three-valued bit-vector abstraction~\cite{schlich2006,noll2008} with fast abstract operations~\cite{onderka2022}, abstracting as
\begin{subequations}
\label{eq:mck_basic}
\begin{align}
    &\gamma^{\text{bit}}(\hat{a}) = \{v \in \mathbb{B} \; | \; (v = 0 \Rightarrow a \neq \text{`1'}) \land (v = 1 \Rightarrow a \neq \text{`0'}) \},\\
    &\hat{S} = \{ \text{`0'}, \text{`1'}, \text{`X'} \}^w, \, \gamma(\hat{s}) = \{ s \in S \; | \; \forall k \in [0, w-1] \; . \; s_k \in \gamma^{\text{bit}}(\hat{s}_k) \}, \label{eq:mck_basic_sub_states}\\
    &\hat{I} = \{  \text{`0'}, \text{`1'}, \text{`X'} \}^y, \, \zeta(\hat{i}) = \{ i \in I \; . \; \forall k \in [0, y-1] \; . \; i_k \in \gamma^{\text{bit}}(\hat{i}_k) \}. \label{eq:mck_basic_sub_inputs}
\end{align}
\end{subequations}
Again, $\hat{s}_0$ is a dummy state with $\gamma(\hat{s}_0) = \{s_0\}$. We rewrite the step function $f$ into an abstract function $\hat{f}^{\text{basic}}: \hat{S} \times \hat{I} \rightarrow \hat{S}$. To formalise the manipulation in Figure \ref{fig:ours}, we use an \emph{input precision function} $\hat{p}_{\hat{q}}: \hat{S} \rightarrow \{0, 1\}^y$ and a \emph{step precision function} $\hat{p}_{\hat{f}}: \hat{S} \rightarrow \{ 0, 1\}^w$, defining the result of $\hat{q}$ and $\hat{f}$ in each bit $k$ by
\begin{subequations}
\label{eq:mck_precision}
\begin{align}
    &(\hat{p}_{\hat{q}}(\hat{s})_k = 0 \Rightarrow \hat{q}(\hat{s})_k = \{ \text{`X'}  \} ) \land (\hat{p}_{\hat{q}}(\hat{s})_k = 1 \Rightarrow \hat{q}(\hat{s})_k = \{ \text{`0'}, \text{`1'} \}), \label{eq:mck_precision_sub_input} \\
    &(\hat{p}_{\hat{f}}(\hat{s})_k = 0 \Rightarrow \hat{f}(\hat{s}, \hat{i})_k = \text{`X'}) \land (\hat{p}_{\hat{f}}(\hat{s})_k = 1 \Rightarrow \hat{f}(\hat{s}, \hat{i})_k = \hat{f}^\textrm{basic}(\hat{s}, \hat{i})_k). \label{eq:mck_precision_sub_step}
\end{align}
\end{subequations}
This allows precise control of the size of the reachable abstract state space. For each $\hat{s} \in \hat{S}$, if $\hat{p}_{\hat{q}}(\hat{s}) = (0)^y$, there is exactly one outgoing transition. Each bit set to $1$ increases that up to a factor of 2. If $\hat{p}_{\hat{f}}(\hat{s})=(0)^w$, there is exactly one outgoing transition to the ``most-decayed'' state $(\text{`X'})^w$.
\end{example}

\subsection{Soundness, Monotonicity, and Completeness}
\label{subsec:characteristics}

\def\coarse{\uparrow}
\def\fine{\downarrow}

In this subsection, we state the requirements sufficient to ensure soundness (the algorithm returns the correct result if it terminates), monotonicity (refinements never lose any information), and completeness (the algorithm always terminates). We sketch the proofs and defer the full proofs to Appendix \ref{app:proofs}.

To intuitively describe the requirements, we formalise the concept of \textit{coverage}. An abstract state~$\hat{s}$ or input~$\hat{i}$ \emph{covers} a concrete $s \in S$ or $i \in I$ exactly when $s \in \gamma(\hat{s})$ or $i \in \zeta(\hat{i})$, respectively, and it \emph{covers} another abstract state $\hat{s}^* \in \hat{S}$ or input $\hat{i}^* \in \hat{I}$ exactly when $\gamma(\hat{s}^*) \subseteq \gamma(\hat{s})$ or $\zeta(\hat{i}^*) \subseteq \zeta(\hat{i})$, respectively.

We want abstraction to preserve the truth value of \textmu-calculus properties in the following sense:
\begin{definition}
    A partial Kripke structure $K^\coarse$ is \emph{sound} with respect to a partial Kripke structure~$K^\fine$ if, for every property $\phi$ of \textmu-calculus over the set of atomic propositions $\mathbb{A}$, it holds that
    \begin{equation}
        \llbracket \phi \rrbracket(K^\coarse) \neq \bot \Rightarrow \llbracket \phi \rrbracket(K^\fine) = \llbracket \phi \rrbracket(K^\coarse).
    \end{equation}
\end{definition}

Intuitively, $K^\coarse$ can contain less information than $K^\fine$, turning some non-$\x$ proposition results to $\x$. No other differences are possible.

To ensure the soundness of Algorithm~\ref{alg:framework}, we use the following requirements. Soundness is ensured with any refinement heuristic as long as they are met.

\begin{definition}
\label{def:abstraction_soundness}
A GA $\hat{G}=(\hat{S}, \hat{s}_0, \hat{I}, \hat{q}, \hat{f}, \hat{L})$ is a \emph{soundness-guaranteeing $(\gamma,\zeta)$-abstraction} of a concrete GA $G=(S, s_0, I, q, f, L)$ iff 
\begin{subequations}
\label{eq:soundness}
\begin{align}
    &\gamma(\hat{s}_0) = \{ s_0 \}, \label{eq:soundness_sub_initial}\\
    &\forall \hat{s} \in \hat{S} \; . \; \forall s \in \gamma(\hat{s}) \; . \; \forall a \in \mathcal{A} \; . \; (\hat{L}(\hat{s}, a) \neq \bot \Rightarrow \hat{L}(\hat{s}, a) = L(s, a)), \label{eq:soundness_sub_labelling}\\
    &\forall (\hat{s}, i) \in \hat{S} \times I \; . \; \exists \hat{i} \in \hat{q}(\hat{s}) \; . \; i \in \zeta(\hat{i}), \label{eq:soundness_sub_input_coverage}\\
    &\forall (\hat{s}, \hat{i}) \in \hat{S} \times \hat{I} \; . \; \forall (s, i) \in \gamma(\hat{s}) \times \zeta(\hat{i}) \; . \; f(s, i) \in \gamma(\hat{f} (\hat{s}, \hat{i})). \label{eq:soundness_sub_step_coverage}
\end{align}
\end{subequations}
\end{definition}

Informally, the four requirements express the following:
\begin{enumerate}[(a)]
\item \textbf{Initial state concretization.} The abstract initial state has exactly the concrete initial state in its concretization.
\item \textbf{Labelling soundness.} Each abstract state labelling must either correspond to the labelling of all concrete states it covers or be unknown.
\item \textbf{Full input coverage.} In every abstract state, each concrete input must be covered by some qualified abstract input.
\item \textbf{Step soundness.} Each result of the abstract step function must cover all results of the concrete step function where its arguments are covered by the abstract step function arguments.

\end{enumerate}

The requirements ensure the soundness of the used abstractions as follows.

\begin{theorem}[Soundness]
\label{th:soundness}
For every generating automaton~$\hat{G}$ and concrete generating automaton~$G$, state concretization function~$\gamma$, and input concretization function~$\zeta$ such that 
 $\hat{G}$ is a soundness-guaranteeing $(\gamma,\zeta)$-abstraction of $G$, the partial Kripke structure $\generatepks(\hat{G})$ is sound with respect to $\generatepks(G)$.
\end{theorem}

\begin{sketch}
Show\footnote{We based the full proofs in Appendix \ref{app:proofs} on the proof sketches given in this section.} that $\{ (s, \hat{s}) \; | \; \hat{s} \in \hat{S} \land s \in \gamma(\hat{s})\}$ is a modal simulation~\cite[p.~408]{dams2018} from $\Gamma(G)$ to $\Gamma(\hat{G})$ due to \eqref{eq:soundness}. Then, $\hat{G}$ is sound wrt. $G$ due to a previous theorem on preservation of \textmu-calculus formulas~\cite[p.~410]{dams2018}.
\end{sketch}

\begin{corollary}
\label{col:soundness}
Assume that the functions { \sc Abstract } and  { \sc Refine } ensure that the generating automaton $(\hat{S}, \hat{s}_0, \hat{I}, \hat{q}, \hat{f}, \hat{L})$ in Algorithm~\ref{alg:framework} is always a soundness-guaranteeing $(\gamma,\zeta)$-abstraction of $(S, s_0, I, q, f, L)$. Then, if the algorithm terminates, its result is correct.
\end{corollary}

\begin{example}
\label{ex:mck_soundness}
Continuing from Example \ref{ex:mck_ga}, \eqref{eq:soundness_sub_initial} is fulfilled trivially. \eqref{eq:soundness_sub_input_coverage} is fulfilled due to \eqref{eq:mck_basic_sub_inputs} and \eqref{eq:mck_precision_sub_input}. From \eqref{eq:mck_precision_sub_step}, it is apparent that
\begin{equation}
    \forall (\hat{s}, \hat{i}) \in (\hat{S}, \hat{I}) \; . \; \gamma(\hat{f}^\textrm{basic}(\hat{s}, \hat{i})) \subseteq \gamma(\hat{f}(\hat{s}, \hat{i})),
\end{equation}
i.e. results of $\hat{f}$ cover results of $\hat{f}^\textrm{basic}$ that abstracts $f$. We carefully implemented the translation of $f$ to $\hat{f}^\textrm{basic}$ and $\hat{L}$ so that\eqref{eq:soundness_sub_labelling} and \eqref{eq:soundness_sub_step_coverage} hold.
\end{example}

\def\rold{}
\def\roldp{\phantom{'}}
\def\rnew{'}

Next, we turn to monotonicity, which ensures no algorithm loop iteration loses information. We give the requirements for the refinement to guarantee it.
\begin{definition}
A generating automaton $(\hat{S}, \hat{s}_0, \hat{I}, \hat{q}, \hat{f}, \hat{L})$ is monotone wrt. $(\gamma,\zeta)$-coverage iff 
\begin{subequations}
\label{eq:ga_monotonicity}
\begin{align}
& \begin{split}
    &\forall (\hat{s}\rold, \hat{s}\rnew, a) \in \hat{S} \times \hat{S} \times \mathcal{A} \; . \; \\
    &\hspace{0.8cm}((\gamma(\hat{s}\rnew) \subseteq \gamma(\hat{s}\rold) \land \hat{L}(\hat{s}\rold, a) \neq \bot) \Rightarrow \hat{L}(\hat{s}\rnew, a) = \hat{L}(\hat{s}\rold, a)),\\
\end{split} \label{eq:labelling_monotonicity}\\
& \begin{split}
    &\forall (\hat{s}\rold, \hat{s}\rnew, \hat{i}\rold, \hat{i}\rnew) \in \hat{S} \times \hat{S} \times \hat{I} \times \hat{I} \; . \; \\
    &\hspace{0.8cm}((\gamma(\hat{s}\rnew) \times \zeta(\hat{i}\rnew) \subseteq \gamma(\hat{s}\rold) \times \zeta(\hat{i}\rold)) \Rightarrow \gamma(\hat{f}(\hat{s}\rnew, \hat{i}\rnew)) \subseteq \gamma(\hat{f}(\hat{s}\rold, \hat{i}\rold)).
\end{split} \label{eq:step_monotonicity}
\end{align}
\end{subequations}
\end{definition}

Informally, we require that each abstract state has at least as much labelling information as each abstract state it covers, and the abstract step function result covers each result produced using covered arguments.

\begin{definition}
\label{def:monotone}
The generating automaton $(\hat{S}, \hat{s}_0, \hat{I}, \hat{q}\rnew, \hat{f}\rnew, \hat{L})$ is a $(\gamma,\zeta)$-monotone refinement of the generating automaton $(\hat{S}, \hat{s}_0, \hat{I}, \hat{q}, \hat{f}, \hat{L})$ iff it is monotone wrt. $(\gamma,\zeta)$-coverage and 
\begin{subequations}
\label{eq:monotonicity}
\begin{align}
    &\forall \hat{s} \in \hat{S} \; . \; \forall \hat{i}\rnew \in \hat{q}\rnew(\hat{s}) \; . \; \exists \hat{i}\rold \in \hat{q}\rold(\hat{s}) \; . \; \zeta(\hat{i}\rnew) \subseteq \zeta(\hat{i}\rold), \label{eq:monotonicity_sub_new_inputs}\\
    &\forall \hat{s} \in \hat{S} \; . \; \forall \hat{i}\rold \in \hat{q}\rold(\hat{s}) \; . \; \exists \hat{i}\rnew \in \hat{q}\rnew(\hat{s}) \; . \;  \zeta(\hat{i}\rnew) \subseteq \zeta(\hat{i}\rold), \label{eq:monotonicity_sub_old_inputs}\\
    &\forall (\hat{s}, \hat{i}) \in \hat{S} \times \hat{I} \; . \; \gamma(\hat{f}\rnew(\hat{s}, \hat{i})) \subseteq \gamma(\hat{f}\rold(\hat{s}, \hat{i})) \label{eq:monotonicity_sub_step_coverage}.
\end{align}
\end{subequations}
\end{definition}
Informally, in addition to the monotonicity wrt. coverage, we also require:
\begin{enumerate}[(a)]
    \item \textbf{New qualified inputs are not spurious.} Each new qualified input is covered by at least one old qualified input.
    \item \textbf{Old qualified inputs are not lost.} Each old qualified input covers at least one new qualified input.
    \item \textbf{New step function covered by old.} The result of the new step function is always covered by the result of the old step function. 
\end{enumerate}
The need for both quantifier combinations in the first two requirements in Equation~\ref{eq:monotonicity} may be surprising. Their violations correspond to transition addition and removal, respectively, which could make a previously non-$\bot$ property $\bot$.

\begin{theorem}[Monotonicity]
\label{th:monotonicity}
If the generating automaton $\hat{G}\rnew$ is a $(\gamma,\zeta)$-mono\-tone refinement of the generating automaton $\hat{G}\rold$, then 
for every \textmu-calculus property $\psi$ for which $\llbracket \psi \rrbracket(\generatepks(\hat{G}\rold)) \neq \bot$, it also holds $\llbracket \psi \rrbracket(\generatepks(\hat{G}\rnew)) \neq \bot$.
\end{theorem}

\begin{sketch}
Show that $\{ (\hat{s}\rnew, \hat{s}\rold) \; | \; \hat{s}\rnew \in \hat{S}\rnew \land s\rold \in S\rold)\}$ is a modal simulation~\cite[p.~408]{dams2018} from $\Gamma(\hat{G}\rnew)$ to $\Gamma(\hat{G}\rold)$ due to \eqref{eq:ga_monotonicity} and \eqref{eq:monotonicity}. Then, $\hat{G}\rold$ is sound wrt. $\hat{G}\rnew$ due to a previous theorem on preservation of \textmu-calculus formulas~\cite[p.~410]{dams2018}.
\end{sketch}

\begin{corollary}
If the update in the loop of Algorithm~\ref{alg:framework} performs a $(\gamma,\zeta)$-monotone refinement of the generating automaton $(\hat{S}, \hat{s}_0, \hat{q}, \hat{f}, \hat{L})$ and for a \textmu-calculus property $\psi$, it held that $\llbracket\psi\rrbracket(\generatepks((\hat{S}, \hat{s}_0, \hat{q}, \hat{f}, \hat{L})))\neq\x$ before the loop iteration, then this is also the case after the iteration.
\end{corollary}

Now we turn to termination. We use the following requirements to ensure that the algorithm makes progress.

\begin{definition}
\label{def:strict_monotone}
The generating automaton $(\hat{S}, \hat{s}_0, \hat{I}, \hat{q}\rnew, \hat{f}\rnew, \hat{L})$ is a \emph{strictly} $(\gamma,\zeta)$-monotone refinement of the generating automaton $(\hat{S}, \hat{s}_0, \hat{I}, \hat{q}, \hat{f}, \hat{L})$ if it is a monotone refinement and either
\begin{align}
    &\exists \hat{s} \in \hat{S} \; . \; \exists \hat{i}\rold \in \hat{q}\rold(\hat{s}) 
\; . \; \forall \hat{i}\rnew \in \hat{q}\rnew(\hat{s}) \; . \; \exists i \in \zeta(\hat{i}\rold) \; . \; i \not\in \zeta(\hat{i}\rnew), \label{eq:input_tightening} \\
\text{or\hspace{0.5em}}
    &\exists (\hat{s}, \hat{i}) \in \hat{S} \times \hat{I} \; . \; \exists s \in \gamma(\hat{f}\rold(\hat{s}, \hat{i})) \; . \; s \not\in \gamma(\hat{f}\rnew(\hat{s}, \hat{i})). \label{eq:step_tightening}
\end{align}
\end{definition}

\begin{example}
\label{ex:mck_monotonicity}
Continuing from Example \ref{ex:mck_soundness}, during each refinement, we set at least one bit of $\hat{p}_{\hat{q}}$ and/or $\hat{p}_{\hat{f}}$ to 1, and prohibit resetting them to 0, fulfilling \eqref{eq:monotonicity}. We implemented $\hat{L}$ so that \eqref{eq:labelling_monotonicity} holds. \eqref{eq:step_monotonicity} holds due to \eqref{eq:mck_precision_sub_step}. As setting bits in $\hat{p}_{\hat{f}}$ does not necessarily mean $\hat{f}$ changes, we set them until $R$ in $\Gamma(\hat{G})$ changes. This means $\hat{q}$ or $\hat{f}$ change as per \eqref{eq:input_tightening} or \eqref{eq:step_tightening}, and the refinement is strictly monotone.
\end{example}

\begin{theorem}[Completeness]
\label{th:completeness}
If $\hat{S}$ and $\hat{I}$ are finite, there is no infinite sequence of generating automata that are soundness-guaranteeing $(\gamma,\zeta)$-abstrac\-tions of some $G$ such that all subsequent pairs in the sequence are strictly $(\gamma,\zeta)$-monotone refinements.
\end{theorem}

\begin{sketch}

Assume such a sequence exists. For each sequence element $\hat{G}$, define a function $M_{\hat{G}}(\hat{s}) = (\{ \zeta(\hat{i}) \; | \;  \hat{i} \in \hat{q}_{\hat{G}}(\hat{s}) \}, \{ (\hat{i} \in \hat{I}, \gamma(\hat{f}_{\hat{G}}(\hat{s}, \hat{i}))) \})$. Show that $M_{\hat{G}}$ must be different for different sequence elements. As there is only a finite number of different functions $M_{\hat{G}}$, the sequence cannot be infinite.

\end{sketch}

\begin{corollary}
\label{cor:completeness}
If $\hat{S}, \hat{I}$ are finite, the functions \Call{Abstract}{} and \Call{Refine}{} in Algorithm~\ref{alg:framework} ensure that $(\hat{S}, \hat{s}_0, \hat{I}, \hat{q}, \hat{f}, \hat{L})$ always is a \emph{soundness-guaranteeing $(\gamma,\zeta)$-abstraction} of $(S, s_0, I, q, f, L)$, and calls of \Call{Refine}{} perform strict $(\gamma,\zeta)$-mono\-tone refinements, then the algorithm returns the correct result in finite time.
\end{corollary}

Fulfilling the requirements of Corollary \ref{cor:completeness} is not trivial. We must exclude the situation when no strict $(\gamma,\zeta)$-monotone refinement is possible any more while a non-$\bot$ result has not yet been reached. We propose Lemma~\ref{lemma:strict} for easier reasoning.

\def\thineq{\kern-0.15em=\kern-0.15em}
\def\thinin{\kern-0.15em\in\kern-0.15em}
\def\thintimes{\kern-0.15em\times\kern-0.15em}

\begin{definition}
    A generating automaton $\hat{G}=(\hat{S}, \hat{s}_0, \hat{I}, \hat{q}, \hat{f}, \hat{L})$ is $(\gamma,\zeta)$-ter\-mi\-nating wrt. a concrete generating automaton $G=(S, s_0, I, q, f, L)$ if it is a $(\gamma,\zeta)$-abstraction of $G$ monotone wrt. $(\gamma,\zeta)$-coverage and
\begin{subequations}
\label{eq:completeness}
\begin{align}
    &\forall (s, \hat{s}) \in S \times \hat{S} \; . \; (\gamma(\hat{s}) = \{ s \} \Rightarrow \hat{L}(\hat{s}) = L(s)), \label{eq:completeness_sub_labelling}\\
    &\forall \hat{s} \in \hat{S} \; . \; \forall \hat{i} \in \hat{q}(\hat{s}) \; . \; \exists i \in I \; . \; \zeta(\hat{i}) = \{i\}, \label{eq:completeness_sub_abstract_inputs}\\
    &\forall \hat{s} \in \hat{S} \; . \; \forall i \in I \; . \; \exists \hat{i} \in \hat{q}(\hat{s}) \; . \; \zeta(\hat{i}) = \{i\}, \label{eq:completeness_sub_concrete_inputs}\\
    &\forall (\hat{s}, \hat{i}, s, i) \thinin \hat{S} \thintimes \hat{I} \thintimes S \thintimes I . ((\gamma(\hat{s}), \zeta(\hat{i})) \thineq (\{ s \}, \{i\}) \Rightarrow \gamma(\hat{f}(\hat{s}, \hat{i})) \thineq \{ f(s,i) \}).
    \label{eq:completeness_sub_step}
\end{align}
\end{subequations}
\end{definition}

\begin{lemma}
\label{lemma:strict}
If $\hat{G}$ is $(\gamma,\zeta)$-terminating wrt. $G$, then for every \textmu-calculus property~$\psi$, $\llbracket \psi \rrbracket(\generatepks(\hat{G})) \kern-0.2em = \kern-0.2em \llbracket \psi \rrbracket(\generatepks(G))$. Furthermore, if $\hat{G}$ is a soundness-guaranteeing $(\gamma,\zeta)$-abstraction of $G$ for which some $(\gamma,\zeta)$-monotone refinement is $(\gamma,\zeta)$-terminating wrt. $G$, then $\hat{G}$ itself is $(\gamma,\zeta)$-terminating wrt. $G$ or the refinement is strict.
\end{lemma}

\begin{sketch}
Show that $\{ (\hat{s}, s) \; | \; \hat{s} \in \hat{S} \land s \in \gamma(\hat{s})\}$ is a modal simulation from $\Gamma(\hat{G})$ to $\Gamma(G)$ due to \eqref{eq:soundness} and \eqref{eq:completeness}. As $\llbracket \phi \rrbracket (\Gamma(G)) \neq \bot$, $\llbracket \psi \rrbracket(\generatepks(\hat{G})) = \llbracket \psi \rrbracket(\generatepks(G))$. Then, consider that if the $\hat{G}$ in the second part is not $(\gamma,\zeta)$-terminating, at least one of \eqref{eq:completeness_sub_labelling}, \eqref{eq:completeness_sub_abstract_inputs}, \eqref{eq:completeness_sub_concrete_inputs}, \eqref{eq:completeness_sub_step} does not hold for $\hat{G}$. Show on a case-by-case basis this is a contradiction for \eqref{eq:completeness_sub_labelling} and forces the refinement to be strict otherwise.
\end{sketch}

\begin{corollary}
\label{cor:strict}
If every generating automaton constructed in Algorithm \ref{alg:framework} fulfils the conditions of Lemma \ref{lemma:strict}, calls to \Call{Refine}{} can always perform a~\emph{strict} $(\gamma,\zeta)$-monotone refinement to a soundness-guaranteeing $(\gamma,\zeta)$-abstraction.
\end{corollary}

\begin{example}
\label{ex:mck_completeness}
Continuing from Example \ref{ex:mck_monotonicity}, we implemented $\hat{L}, \hat{f}^\textrm{basic}$ so that they fulfil \eqref{eq:completeness_sub_labelling} and \eqref{eq:completeness_sub_step}. Since \eqref{eq:completeness_sub_abstract_inputs} and \eqref{eq:completeness_sub_concrete_inputs} hold due to \eqref{eq:mck_precision_sub_input}, $\hat{G}^*$ obtained by $\hat{p}_{\hat{q}}=(1)^y, \hat{p}_{\hat{f}}=(1)^w$ is $(\gamma,\zeta)$-terminating. Inspecting \eqref{eq:monotonicity}, $\hat{G}^*$ is monotone wrt. all other applicable GA, and Corollary \ref{cor:strict} holds. As soundness was already ensured in Example \ref{ex:mck_soundness} and $\hat{S}$, $\hat{I}$ are finite by definition, Corollary \ref{cor:completeness} holds.
\end{example}

%% file: sources/04_choices.tex
\section{Making Reasonable Refinement Choices}
\label{sec:choices}

A sound, monotone, and complete instantiation of the framework can e.g. refine randomly while fulfilling the requirements from Subsection \ref{subsec:characteristics}. However, for fast verification, refinements must be chosen carefully. We will discuss how reasonable refinements can be made, inspired by the \texttt{FindFailure} algorithms to find the cause of an unknown result of verification of CTL \cite[p.~551-552]{shoham2004} and game-based verification of \textmu-calculus \cite[p.~1145]{grumberg2007}. 

The goal is to find the root cause of an unknown verification result, an unknown atomic labelling in an abstract state reached through the abstract state space, which we will also call the \emph{culprit}. For this, we extend the semantics of three-valued \textmu-calculus on PKS, a simplification of three-valued \textmu-calculus on KMTS~\cite{grumberg2007,shoham2008}, so that it returns the culprit instead of the unknown value $\bot$. For simplicity, we base this definition on the standard syntax of \textmu-calculus in positive normal form. Assuming, in addition to the set of atomic labellings $\mathbb{A}$, a~set of \textmu-calculus variables $\mathbb{V}$, a formula then has the form
\begin{equation}
    \phi ::= a \; | \; \lnot a \; | \; Z \; | \; \phi \land \phi \; | \; \phi \lor \phi \; | \; \langle \phi \rangle \; | \; [\phi] \; | \; \mu Z \; . \; \phi \; | \; \nu Z \: . \phi
\end{equation}
where $a \in \mathbb{A}$, $Z \in \mathbb{V}$. The language of \textmu-calculus then consists of all closed-form formulas given by this definition. We also assume variables are well-bound, i.e. every variable must be bound by $\mu$ or $\nu$ exactly once.

\begin{definition}
Given a GA $\hat{G}=(\hat{S}, \hat{s}_0, \hat{I}, \hat{q}, \hat{f}, \hat{L})$, an \emph{environment} is a function 
$\rho: \mathbb{V} \times \hat{S} \rightarrow \{0, 1\} \cup C$, where $C$ 
is the set of pairs $(\mathcal{S}, a)$ with $\mathcal{S}$ being a non-empty sequence of states from $\hat{S}$ and $a \in \mathbb{A}$ an atomic labelling.

Now let $\hat{R}$ be the transition relation of the PKS~$\Gamma(\hat{G})$. We define the extended semantics of \textmu-calculus recursively according to its syntax, returning an environment as follows:
\begin{align*}
&\llbracket a \rrbracket ::= 
\begin{cases}
\hat{L}(\hat{s}, a) & \text{if } \hat{L}(\hat{s}, a) \neq \bot,\\
((\hat{s}), a) & \text{otherwise,}
\end{cases} \hspace{3em} 
\llbracket \lnot a \rrbracket ::= 
\begin{cases}
1 - \hat{L}(\hat{s}, a) & \text{if } \hat{L}(\hat{s}, a) \neq \bot,\\
((\hat{s}), a) & \text{otherwise,}
\end{cases}\\
&\llbracket \phi \land \psi \rrbracket ::= \begin{cases}
1 & \text{if } \llbracket \phi \rrbracket = 1 \text{ and } \llbracket \psi \rrbracket = 1,\\
0 & \text{if } \llbracket \phi \rrbracket = 0 \text{ or } \llbracket \psi \rrbracket = 0,\\ c\in C \text{ s.t. } c=\llbracket\phi\rrbracket \text{ or } c=\llbracket\psi\rrbracket
& \text{otherwise,}
\end{cases}\\
&\llbracket \langle \phi \rangle \rrbracket::= \begin{cases} 1, \text{ if there is a } \hat{s}^+ \text{ with } (\hat{s},\hat{s}^+) \in \hat{R} \text{ and  } \llbracket \phi \rrbracket_{s^+} = 1,\\
0, \text{ if for all } \hat{s}^+, (\hat{s},\hat{s}^+) \in \hat{R} \text{ implies }\llbracket \phi \rrbracket_{s^+} = 0,\\
((\hat{s}, \hat{s}_1, \dots), a) \text{ s.t. } (\hat{s}, \hat{s}_1)\in\hat{R} \text{ and } \llbracket \phi \rrbracket_{s_1}=((\hat{s}_1, \dots), a),   \text{otherwise,}
\end{cases}\\
&\llbracket Z \rrbracket ::= \rho(Z), \hspace{11.5em} \llbracket \mu Z \; . \; \phi \rrbracket ::= \text{lfp}(\phi).
\end{align*}
The least fixed point lfp is defined as usual for the ordering $0 < \bot < 1$, but retains the element of $C$ previously obtained when the environmental value remains $\bot$ between subsequent iterations. The operators $\phi \lor \phi$, $[\phi]$, and $\nu Z \: . \phi$ are defined correspondingly to their duals.
\end{definition}

As already mentioned, in the case where the result is an element of $C$, we will call it the \emph{culprit}. The definition is not unique, allowing implementation choice of which culprit to use for refinement. With any choice, the culprit describes a path through the state space ending with a labelling that contributes to the verification result $\llbracket \phi \rrbracket_{\hat{s}_0}$ being unknown.

\textbf{Using the culprit for refinement.} If the instance of our framework fulfils the conditions from Corollary \ref{cor:completeness} then there clearly is at least one application of $\hat{f}$ on the culprit's path that makes the subsequent state cover more than one state, since $\gamma(\hat{s}_0) = \{s_0\}$ and \eqref{eq:completeness_sub_labelling} forbids spuriously unknown labellings. As such, we can decide to only refine $\hat{f}$ from the preceding $\hat{s}$ of such applications. This avoids refinements that definitely would not impact any culprit, unnecessarily increasing the size of the abstract state space. This idea can be extended to the instance details, not refining e.g. inputs that provably do not have an impact on the culprit labelling. The choice of culprit and actual refinement can be fine-tuned according to the typical systems under verification in order to achieve good performance.

\begin{example}
\label{ex:choices}
In \mck{}, we model-check and obtain the culprit in the spirit of the above discussion. To ensure deterministic behaviour, we prefer the left culprit for operators $\phi \land \phi, \phi \lor \phi$ and the culprit continuing with the state with the smallest unique identification number for next-state operators $\langle \phi \rangle, [\phi]$.

For better performance, we use incremental model-checking where we construct a history of environment updates and only update a part of the history after refinement, if possible. To avoid storing the elements of $C$, where paths can be problematically long, we instead store the evaluation choices made to obtain the given history point from previous history points and atomic labellings. This allows us to reconstruct the culprit if necessary by walking back through history.

After obtaining the culprit $((\hat{s}_0, \hat{s}_1, \dots, \hat{s}_{n-1}, \hat{s}_n), a)$, we  compute a cone of influence on $a$ being unknown in $\hat{s}_n$. For each state $\hat{s} \in \{\hat{s}_0, \hat{s}_1, \dots, \hat{s}_{n-1}\}$, we determine candidate bit positions $k \in [0, y-1]$ where $\hat{p}_{\hat{q}}(\hat{s})_k = 0$ and it is possible this has contributed to $a$ being unknown in $\hat{s}_n$, and likewise for $k \in [0, w-1]$ where $\hat{p}_{\hat{f}}(\hat{s})_k = 0$. We then choose the state in which to refine $\hat{p}_{f}$ or $\hat{p}_{q}$ deterministically based on the following heuristic:
\begin{itemize}
    \item If there are candidate bits for refinement of $\hat{p}_{f}$, we choose to refine $\hat{p}_{f}$ in the last state on the culprit path where there is some candidate.
    \item Otherwise, we pre-select only the states containing candidates that have the highest level of indirection (indirection rises for indices to array reads and writes that impact the culprit labelling). We then choose to refine $\hat{p}_{q}$ using the candidates on the last pre-selected state on the culprit path.
\end{itemize}
From multiple candidate bits for the refinement in the chosen state, we choose the input variable deterministically. From multiple bit candidates in a bit-vector variable, we refine the most significant bit.

\end{example}

%% file: sources/05_evaluation.tex
\section{Implementation and Experimental Evaluation}
\label{sec:evaluation}

We implemented an instance of our proposed framework in our free and open-source tool \mck{}. In this section, we discuss our implementation choices and show that our framework is able to mitigate exponential explosion.

\subsection{Implementation}
\label{subsec:implementation}
Our focus is verification of machine-code systems, where the system is composed of a processor and a machine-code program it executes. However, there is no special handling for typical features of machine-code systems such as the Program Counter, call stack, etc., allowing \mck{}  to support verification of arbitrary finite-state digital systems against properties in propositional \textmu-calculus.

\textbf{Systems.} Systems to be verified are described in a subset of the Rust programming language. The systems can be comprised of bit-vectors that support standard arithmetical, logical, shift, and extension operations and bit-vector arrays that support indexing.

\textbf{Properties.} Properties are described in a Rust-style format that allows \textmu-calculus and Computation Tree Logic properties to be expressed. Atomic propositions are formed by equality and comparison operations between variables of the system. CTL and \textmu-calculus can be freely mixed, and CTL operators are converted to \textmu-calculus before model-checking. For example,
the CTL property \textbf{AG}$[\texttt{value} = 0]$ is translated to $\nu Z. (\texttt{value} = 0) \land [Z]$.

\textbf{Translating the systems.} The system description, corresponding to the concrete GA under verification, is compiled together with \mck{} core to form the system verifier, and the description is automatically translated \cite{onderka2024meco} to representations that serve as an abstract simulator and a cone-of-influence calculator. No other interaction between the system and framework is needed.

\textbf{Abstraction.} Three-valued bit-vector abstraction is used for bit-vectors, with fast abstract operations \cite{onderka2022}. Arrays are abstracted using a version of sparse representation where elements that are not stored have the value of the nearest stored element with a lesser index.

\textbf{Building the state space.} The state space is built incrementally using the abstract analogue of the system as the abstract simulator. Abstract state data and a sparse transition graph are retained throughout verification, with garbage collection of abstract states that are no longer reachable.

\textbf{Model-checking.} As noted in Example \ref{ex:choices}, we use incremental model-checking of \textmu-calculus. For reassurance, once the final verification result is obtained, it is double-checked non-incrementally.

\textbf{Verification settings.} With default verification settings, inputs are initially unknown but decay is not used, i.e. $\hat{p}_{\hat{q}}(\hat{s}) = (0)_{k=0}^{y-1}$, $\hat{p}_{\hat{f}}(\hat{s}) = (1)_{k=0}^{w-1}$ from all states $\hat{s} \in \hat{S}$. It is possible to enable decay ubiquitously, i.e. $\hat{p}_{\hat{f}}(\hat{s}) = (0)_{k=0}^{w-1}$, but we found this to be slower for machine-code systems due to the need to refine each generated state. A naïve strategy with $\hat{p}_{\hat{q}}(\hat{s}) = (1)_{k=0}^{y-1}$ corresponding to explicit-state verification without abstraction would immediately result in hopeless explosion of the state space due to the amount of initially uninitialised memory and read inputs in machine code systems.

\subsection{Evaluation Setup}
We evaluated\footnote{Detailed results and reproduction data, including the benchmark set, are available at \url{https://doi.org/10.5281/zenodo.17167265}.} verification of \textmu-calculus properties on machine-code systems for the AVR ATmega328P microcontroller that we described for use in \mck{} according to the datasheet \cite{ATmega328pDatasheet} and instruction set reference \cite{avrinstructionset}.

\textbf{Benchmark set.} As we are not aware of any publicly available and appropriately licensed benchmark sets for formal verification of embedded 8-bit microcontrollers, we used our own set of benchmarks based on machine-code programs for ATmega328P. Our set of benchmarks currently contains 16 programs: 6 programs where we expect a violation of some system-inherent property (e.g. forbidden instruction), 5 simple programs, and 5 more complicated programs, four of which are variations of a simplified version of a program for calibration of a Voltage-Controlled Oscillator used in real life. For C programs, we benchmark the machine code obtained with debug and release builds separately.

\textbf{Evaluation setup.} The evaluation was performed on a Linux machine with an Intel Core i9-12900 processor with 128 GB of RAM. The tool was built with Rust 1.83.0 in release configuration. The default verification strategy was used.

\textbf{Evaluated properties.} We evaluated up to 8 properties on each program\footnote{To ensure the implementation is not faulty, we also test the properties on simple non-machine-code systems in our test suite, e.g. whether the verifier distinguishes between \textbf{AFG}[$w$] and \textbf{AF}[\textbf{AG}[$w$]] using the standard example \cite[p.~65]{Piterman2018}.}. Using $w$ for a propositional formula on atomic properties that varies depending on the program, the properties included, among others:
\begin{itemize}
    \item Inherent property defined in the processor description, determining if a system with given machine-code is permissible, containing no calls to unimplemented instructions, peripherals, etc. \textbf{AG}[$w$], i.e. a safety property.
    \item Recovery property (as discussed in Example \ref{ex:system}), checking whether the system can be recovered from every reachable state to the program loop start with initial output values. \textbf{AG}[\textbf{EF}[$w$]] in CTL. Not a linear-time property.
    \item Program counter (PC) remains within the main program loop in the main function infinitely often on all paths. \textbf{FG}[$w$] in LTL, $\mu X \; . 
    \; \nu Y \; . \; [X] \lor (w \land [Y])$ in \textmu-calculus \cite[p.~3133]{cranen2011}. Not present in CTL.
    \item PC is never at the start of the main loop during even times (wrt. instructions). $\nu X \; . 
    \; w \land [\, [X] \,]$ in \textmu-calculus, not present in CTL, LTL, nor CTL*.
    \item The stack pointer stays above a given value, preventing stack overflow problems. \textbf{AG}[$w$], i.e. a safety property.
\end{itemize}

\subsection{Evaluation Results}
We performed 126 measurements in total on the evaluated machine-code systems and properties, all of which resulted in the property being verified. All results matched our expectations. We will discuss a limited number of the more interesting measurements in detail, demonstrating the applicability of the framework.

\begin{table}[t]
\caption{Selected measurements of verification of the calibration program compiled in debug configuration. Asterisks mark measurements where the inherent property was verified first and verification of the given property progressed from that state space.}
\label{table:measurements_calibration}
\centering
\fontsize{7}{10}\selectfont
{\tabcolsep=0.06cm
\renewcommand{\arraystretch}{1.11}
\begin{tabular}{|c|c|c|c|c|c|c|c|}
\hline
\textbf{Variant}   & \textbf{Property name} & \textbf{Holds} & \textbf{Refin.} & \textbf{States} & \textbf{Transitions} & \textbf{CPU t. {[}s{]}} & \textbf{Mem. {[}MB{]}} \\ \hline
Original    & Inherent           & \cmark & 513 & 13059 & 13573 & 17.09  & 87.39  \\
Original    & Recovery           & \xmark & 513 & 13059 & 13573 & 19.88  & 111.34 \\
Original    & Infinitely often   & \cmark & 513 & 13059 & 13573 & 672.05 & 160.58 \\
Original    & Infinitely often*  & \cmark & 513 & 13059 & 13573 & 17.98  & 115.34 \\
Original    & Even non-starts    & \xmark & 0   & 17    & 18    & $<$0.01& 3.78   \\
Original    & Stack above 0x08FD & \cmark & 513 & 13059 & 13573 & 17.08  & 87.30  \\
Original    & Stack above 0x08FE & \xmark & 0   & 17    & 18    & $<$0.01& 3.94   \\ \hline
Fixed       & Inherent           & \cmark & 513 & 13059 & 13573 & 17.21  & 87.60  \\
Fixed       & Recovery           & \cmark & 513 & 13059 & 13573 & 29.16  & 112.02 \\
Fixed       & Infinitely often   & \cmark & 513 & 13059 & 13573 & 672.19 & 161.65 \\
Fixed       & Infinitely often*  & \cmark & 513 & 13059 & 13573 & 18.07  & 115.66 \\
Fixed       & Even non-starts    & \xmark & 0   & 17    & 18    & $<$0.01& 3.72   \\
Fixed       & Stack above 0x08FD & \cmark & 513 & 13059 & 13573 & 17.2   & 87.56  \\
Fixed       & Stack above 0x08FE & \xmark & 0   & 17    & 18    & $<$0.01& 3.98   \\ \hline
Comp. orig. & Inherent           & \cmark & 771 & 17330 & 18102 & 54.43  & 125.83 \\
Comp. orig. & Recovery           & \xmark & 770 & 17333 & 18104 & 65.28  & 159.20 \\
Comp. orig. & Infinitely often*  & \cmark & 771 & 17330 & 18102 & 73.29  & 173.08 \\ \hline
Comp. fixed & Inherent           & \cmark & 771 & 17330 & 18102 & 51.63  & 125.71 \\
Comp. fixed & Recovery           & \cmark & 771 & 17330 & 18102 & 70.74  & 159.52 \\
Comp. fixed & Infinitely often*  & \cmark & 771 & 17330 & 18102 & 64.2   & 173.04 \\ \hline
\end{tabular}
}
\vspace{-0.5em}
\end{table}

\textbf{Calibration.} We used a simplified version of a program for calibration of a Voltage-Controlled Oscillator (VCO) using binary search. We previously used the program in a real device, where the VCO frequency was adjusted by a digital potentiometer based on the measured frequency of the produced wave. For verification, we replaced the interactions with the digital potentiometer and frequency measure by I/O writes and reads, leaving the core algorithm unchanged.

Verifying using \mck{}, we found a bug in the program using the recovery property: the output value could never recover to zero after an iteration concludes, since the least significant bit was never set to zero during the calibration, causing a calibration precision loss. This bug would have been tricky to find using source-code verification, as the problem occurred in peripheral manipulation. Linear-time properties of forms \textbf{F}[$w$] or \textbf{FG}[$w$] would not find the bug, as the calibration command may never be given. After fixing the bug, the recovery property holds, as seen in Table \ref{table:measurements_calibration}.

Despite only the simple three-valued bit-vector abstraction used, the final sizes of abstract state spaces are reasonable. The refinement guidance is remarkably well-behaved, arriving at the same state space for the properties where a fast decision was not possible, despite the non-inherent properties having no knowledge about e.g. illegal instructions. Verification using the infinitely-often property is notably slow due to the need to model-check with non-trivial nested quantifiers, but this slowdown is effectively mitigated by using the inherent property first.

\textbf{Need for abstraction.} To show Input-based TVAR holds up where explicit-state verification would be completely infeasible, we created complicated calibration variants where a $64$-bit value is read during initialisation, and unrelated 8-bit volatile reads are performed while doing the calibration, ensuring there are more than $2^{80}$ concrete states even if not considering uninitialised memory. As seen in Table \ref{table:measurements_calibration}, \mck{} verifies with at most approx. factor-of-4 slowdown.

\begin{table}
\caption{Selected measurements of verification of the factorial program. Everything is as expected. The recovery property does not hold as the factorial program never outputs zero after the first computation. The infinitely-often property does not hold due to calls to another function from the main function affecting the Program Counter.}
\label{table:measurements_factorial}
\centering
\fontsize{7}{10}\selectfont
{\tabcolsep=0.078cm
\renewcommand{\arraystretch}{1.25}
\begin{tabular}{|c|c|c|c|c|c|c|c|}
\hline
\textbf{Program}   & \textbf{Property name} & \textbf{Holds} & \textbf{Refin.} & \textbf{States} & \textbf{Transitions} & \textbf{CPU t. {[}s{]}} & \textbf{Mem. {[}MB{]}} \\ \hline
Debug   & Inherent           & \cmark & 576 & 51595 & 52940 & 50.4  & 359.79 \\
Debug   & Recovery           & \xmark & 576 & 51595 & 52940 & 76.48 & 469.88 \\
Debug   & Infinitely often   & \xmark & 6   & 1396  & 1411  & 5.55  & 259.89 \\
Debug   & Even non-starts    & \cmark & 576 & 51595 & 52940 & 46.96 & 360.40 \\
Debug   & Stack above 0x08DD & \cmark & 576 & 51595 & 52940 & 52.77 & 359.74 \\
Debug   & Stack above 0x08DE & \xmark & 3   & 748   & 756   & 0.16  & 31.58  \\ \hline
Release & Inherent           & \cmark & 45  & 4272  & 4378  & 0.78  & 46.84  \\
Release & Recovery           & \xmark & 45  & 4272  & 4378  & 1.11  & 54.58  \\
Release & Infinitely often   & \xmark & 6   & 1006  & 1021  & 2.93  & 148.15 \\
Release & Even non-starts    & \xmark & 3   & 550   & 558   & 0.04  & 17.92  \\
Release & Stack above 0x08FB & \cmark & 45  & 4272  & 4378  & 0.78  & 46.79  \\
Release & Stack above 0x08FC & \xmark & 0   & 20    & 21    & $<0.01$& 3.72   \\ \hline
\end{tabular}
}
\end{table}

\textbf{Factorial.} We implemented a program in C which computes the factorial of an input recursively. In the compiled machine code, the computation remained recursive in the debug build but was optimised to an iterative version in the release build. We verified properties including maximal stack sizes for both versions, as seen in Table \ref{table:measurements_factorial}. Regarding the state space size, we noticed that the call return locations remain known even after the stack is popped and they are no longer relevant, unnecessarily growing the state space. We believe a smarter, selective decay of step precision could alleviate these problems.

\textbf{Comparison with other tools.} While there has been a multitude of previous machine-code verifiers \cite[p.~29-37]{onderka2025}, we are not aware of any that support full \textmu-calculus. Our work has been inspired by the research on the tool [mc]square / Arcade.\textmu C \cite{schlich2006,delayed,noll2008,schlich2010,reinbacher2009,Gue14}, which could use strategies including Delayed Nondeterminism \cite{noll2008}. However, the tool has been discontinued and is not publicly available. It only supported proving LTL/ACTL properties when abstraction was used, and needed system-based hints for useful abstraction. As for more recent work, verifiers based on theorem proving have been introduced, with the Serval~\cite{nelson2019,serval} tool supporting automatic verification but not programs with loops, and Islaris~\cite{sammler2022,islaris} requiring manual proof support but supporting loops. Both only support verification of safety properties.  In contrast, \mck{} supports fully automatic verification of \textmu-calculus properties on machine-code programs that include loops thanks to the introduced framework.

%% file: sources/06_conclusion.tex
\section{Conclusion}
\label{sec:conclusion}
We presented a novel input-based Three-Valued Abstraction Refinement (TVAR) framework for formal verification of \textmu-calculus properties using abstraction refinement and gave requirements for soundness, monotonicity, and completeness. Our framework can verify properties not verifiable using Counterexample-guided Abstraction Refinement or (Generalized) Symbolic Trajectory Evaluation, eliminating verification blind spots. Compared to previous TVAR frameworks, our framework does not use modal transitions, allowing monotonicity without formalism complications. We implemented the framework in our tool \mck{}, which can verify machine-code systems while mitigating exponential explosion.

%% file: sources/a_proofs.tex
\section{Proofs of Soundness, Monotonicity, and Completeness}
\label{app:proofs}

\input{sources/c_a_modal}

\input{sources/c_b_soundness}
\input{sources/c_c_monotonicity}
\input{sources/c_d_completeness}
\input{sources/c_e_lemma}

%% file: sources/c_a_modal.tex
\def\coarse{\uparrow}
\def\fine{\downarrow}
\def\kerndot{\kern-0.3em}
\def\kerncomma{\kern-0.25em}

In this appendix, we will prove Theorems~\ref{th:soundness}, \ref{th:monotonicity}, and \ref{th:completeness}, and Lemma \ref{lemma:strict}. Since directly proving on \textmu-calculus is cumbersome, we will typically prove the theorems by showing their requirements imply that some PKS is sound wrt. another PKS as per Definition~\ref{def:abstraction_soundness} using a well-known property of modal simulation~\cite[p.~408-410]{dams2018}, which we simplified from KMTS to PKS by setting $R=R_{may}=R_{must}$.

\begin{definition}
\label{def:modal_simulation}
Let $K^\fine=(S^\fine, S_0^\fine, R^\fine, L^\fine)$ and $K^\coarse=(S^\coarse, S_0^\coarse, R^\coarse, L^\coarse)$ be PKS over the set of atomic propositions $\mathcal{A}$. Then $H \subseteq S^\fine \times S^\coarse$ is a modal simulation from $K^\fine$ to $K^\coarse$ if and only if all of the following hold,

\begingroup
\thinmuskip=0.75\thinmuskip
\medmuskip=0.75\medmuskip
\thickmuskip=0.75\thickmuskip
\begin{subequations}
\label{eq:modal_simulation}
\begin{align}
    &\forall (s^\fine\kerncomma, s^\coarse) \in H \, . \, \forall a \in \mathcal{A} \; . \; (L^\coarse(s^\coarse, a) \neq \bot \Rightarrow L^\fine(s^\fine, a) = L^\coarse(s^\coarse, a)), \label{eq:modal_simulation_sub_labelling}\\
    &\forall (s^\fine\kerncomma, s^\coarse) \in H \, . \, \forall s'^\fine \in S^\fine\kerndot \,. \, (R^\fine(s^\fine\kerncomma, s'^\fine) \Rightarrow \exists s'^\coarse \in S^\coarse\kerndot \,. \, (R^\coarse(s^\coarse\kerncomma, s'^\coarse) \land H(s'^\fine\kerncomma, s'^\coarse))), \label{eq:modal_simulation_sub_fine_to_coarse}\\
    &\forall (s^\fine\kerncomma, s^\coarse) \in H \, . \, \forall s'^\coarse \in S^\coarse\kerndot \,. \, (R^\coarse(s^\coarse\kerncomma, s'^\coarse) \Rightarrow \exists s'^\fine \in S^\fine\kerndot \,.\, (R^\fine(s^\fine\kerncomma, s'^\fine) \land H(s'^\fine\kerncomma, s'^\coarse))). \label{eq:modal_simulation_sub_coarse_to_fine}
\end{align}
\end{subequations}
\endgroup
\end{definition}

Informally, $H$ relates the states of $K^\fine$ and $K^\coarse$ so that the states in the fine $K^\fine$ preserve all known labellings of their related states from the coarse $K^\coarse$, and transitions are respected: for every pair of related states in $H$, each transition from one element of the pair must correspond to at least one transition from the other element with the endpoints related by $H$.

\begin{definition}
\label{def:pks_modal_simulation}
    A PKS $K^\fine$ with initial state $s_0^\fine$ is modal-simulated by a PKS~$K^\coarse$ with initial state $s_0^\coarse$, denoted $K^\fine \preceq K^\coarse$, if there is a modal simulation~$H$ from $K^\fine$ to $K^\coarse$ and furthermore,
\begin{equation}
\label{eq:pks_modal_simulation}
    \forall s_0^\fine \in S_0^\fine \; . \; \exists s_0^\coarse \in S_0^\coarse \; . \; H(s_0^\fine, s_0^\coarse).
\end{equation}
\end{definition}

\begin{property}
\label{prop:soundness}
    $K^\coarse$ is sound wrt. $K^\fine$  if $K^\fine \preceq K^\coarse$~\cite[p.~410]{dams2018} (original argument by Huth et al. \cite[p.~161]{huth2001}).
\end{property}

To simplify the proofs further, we first prove a criterion for generating automata that ensures their corresponding Kripke structures are sound.

\begin{lemma}
\label{lemma:lemma}
 The PKS~$\MakePKS{S^\coarse, s_0^\coarse, q^\coarse, f^\coarse, L^\coarse}$ is sound wrt. $\MakePKS{S^\fine, s_0^\fine, q^\fine, f^\fine, L^\fine}$  if there exists a relation $H \subseteq S^\fine \times S^\coarse$ where
\begin{subequations}
\label{eq:lemma}
\begin{align}
    &(s_0^\fine, s_0^\coarse) \in H, \label{eq:lemma_sub_initial}\\
    &\forall (s^\fine, s^\coarse) \in H \; . \; \forall a \in \mathcal{A} \; . \; (L^\coarse(s^\coarse, a) \neq \bot \Rightarrow L^\fine(s^\fine, a) = L^\coarse(s^\coarse, a)), \label{eq:lemma_sub_labelling}\\
    &\forall (s^\fine, s^\coarse) \in H \; . \; \forall i^\fine \in q^\fine(s^\fine) \; . \; \exists i^\coarse \in q^\coarse(s^\coarse) \; . \; H(f^\fine(s^\fine, i^\fine), f^\coarse(s^\coarse, i^\coarse)), \label{eq:lemma_sub_fine_to_coarse}\\
    &\forall (s^\fine, s^\coarse) \in H \; . \; \forall i^\coarse \in q^\coarse(s^\coarse) \; . \; \exists i^\fine \in q^\fine(s^\fine) \; . \; H(f^\fine(s^\fine, i^\fine), f^\coarse(s^\coarse, i^\coarse)). \label{eq:lemma_sub_coarse_to_fine}
\end{align}
\end{subequations}
\end{lemma}

\begin{proof}
Let $K^\coarse=\MakePKS{S^\coarse, s_0^\coarse, q^\coarse, f^\coarse, L^\coarse}$ and $K^\fine=\MakePKS{S^\fine, s_0^\fine, q^\fine, f^\fine, L^\fine}$.
We assume a relation $H \subseteq S^\fine \times S^\coarse$ satisfying \eqref{eq:lemma_sub_initial} to \eqref{eq:lemma_sub_coarse_to_fine} and use it to prove $K^\fine \preceq K^\coarse$ which implies soundness due to Property~\ref{prop:soundness}. Due to \eqref{eq:lemma_sub_initial}, \eqref{eq:pks_modal_simulation} holds, and it remains to prove \eqref{eq:modal_simulation}. \eqref{eq:modal_simulation_sub_labelling} directly follows from \eqref{eq:lemma_sub_labelling}. To prove \eqref{eq:modal_simulation_sub_fine_to_coarse} and \eqref{eq:modal_simulation_sub_coarse_to_fine}, we respectively expand $R^\fine$ and $R^\coarse$ according to Definition~\ref{def:generatepks} to
\begin{subequations}
\label{eq:lemma_1}
\begin{align}
    &\forall (s^\fine, s^\coarse) \in H \; . \; \forall s'^\fine \in S^\fine \; . \;  \label{eq:lemma_1_sub_fine_to_coarse}\\
    &\hspace{0.25cm} ((\exists i^\fine \in q^\fine(s^\fine) \; . \;  s'^\fine = f^\fine(s^\fine, i^\fine)) \Rightarrow \exists s'^\coarse \in S^\coarse \; . \; (R^\coarse(s^\coarse, s'^\coarse) \land H(s'^\fine, s'^\coarse))),\notag\\
    &\forall (s^\fine, s^\coarse) \in H \; . \; \forall s'^\coarse \in S^\coarse \; . \; \label{eq:lemma_1_sub_coarse_to_fine}\\
    &\hspace{0.25cm} ((\exists i^\coarse \in q^\coarse(s^\coarse) \; . \;  s'^\coarse = f^\coarse(s^\coarse, i^\coarse)) \Rightarrow \exists s'^\fine \in S^\fine \; . \; (R^\fine(s^\fine, s'^\fine) \land H(s'^\fine, s'^\coarse))).\notag
\end{align}
\end{subequations}
We move the $i^\fine, i^\coarse$ quantifiers out of the implication, negating them due to moving out of antecedent. $s'^\fine$ and $s'^\coarse$ must be equal to $f^\fine(s^\fine, i^\fine)$ and $f^\coarse(s^\coarse, i^\coarse)$, respectively, so we replace them and eliminate the quantifiers, obtaining
\begin{subequations}
\label{eq:lemma_2}
\begin{align}
    &\forall (s^\fine, s^\coarse) \in H . \; \forall i^\fine \in q^\fine(s^\fine) \; . \; \exists s'^\coarse \in S^\coarse \; . \; (R^\coarse(s^\coarse, s'^\coarse) \land H(f^\fine(s^\fine, i^\fine), s'^\coarse)),\\
    &\forall (s^\fine, s^\coarse) \in H . \; \forall i^\coarse \in q^\coarse(s^\coarse) \; . \; \exists s'^\fine \in S^\fine \; . \; (R^\fine(s^\fine, s'^\fine) \land H(s'^\fine, f^\coarse(s^\coarse, i^\coarse))).
\end{align}
\end{subequations}
We then respectively insert the definition of $R^\coarse$ and $R^\fine$, pull the $i^\coarse, i^\fine$ quantifiers outside, and eliminate the $s'^\coarse, s'^\fine$ variables, obtaining
\begin{subequations}
\label{eq:lemma_3}
\begin{align}
    &\forall (s^\fine, s^\coarse) \in H \; . \; \forall i^\fine \in q^\fine(s^\fine) \; . \; \exists i^\coarse \in q^\coarse(s^\coarse) \; . \; H(f^\fine(s^\fine, i^\fine), f^\coarse(s^\coarse, i^\coarse)),\\
    &\forall (s^\fine, s^\coarse) \in H \; . \; \forall i^\coarse \in q^\coarse(s^\coarse) \; . \; \exists i^\fine \in q^\fine(s^\fine) \; . \; H(f^\fine(s^\fine, i^\fine), f^\coarse(s^\coarse, i^\coarse)),
\end{align}
\end{subequations}
which correspond to the assumed \eqref{eq:lemma_sub_fine_to_coarse} and \eqref{eq:lemma_sub_coarse_to_fine}. \qed

\end{proof}

%% file: sources/c_b_soundness.tex
\subsection{Proof of Soundness}
\label{app_subsec:soundness}

\begin{proof}[Theorem~\ref{th:soundness}]
Assume generating automata~$\hat{G}=(\hat{S}, \hat{s}_0, \hat{I}, \hat{q}, \hat{f}, \hat{L})$ and~$G=(S, s_0, I, q, f, L)$, state concretization function~$\gamma$, and input concretization function~$\zeta$ such that 
 $\hat{G}$ is a soundness-guaranteeing $(\gamma,\zeta)$-abstraction of $G$.
Our goal is to prove that $\MakePKS{\hat{G}}$ is sound wrt. $\MakePKS{G}$.  For this, we prove that 
\begin{equation}
\label{eq:soundness_modal_simulation}
    H = \{ (s, \hat{s}) \in S \times \hat{S} \; | \; s \in \gamma(\hat{s}) \}
\end{equation}
satisfies the conditions of Lemma~\ref{lemma:lemma} which implies that $\MakePKS{\hat{G}}$ is sound wrt. $\MakePKS{G}$. 
Conditions~\eqref{eq:lemma_sub_initial} and \eqref{eq:lemma_sub_labelling} hold due to \eqref{eq:soundness_sub_initial} and \eqref{eq:soundness_sub_labelling}, respectively. We use the definition of $H$ and the fact that $G$ is concrete to rewrite \eqref{eq:lemma_sub_fine_to_coarse} and \eqref{eq:lemma_sub_coarse_to_fine} to
\begin{subequations}
\label{eq:soundness_1}
\begin{align}
    &\forall \hat{s} \in \hat{S} \; . \; \forall s \in \gamma(\hat{s}) \; . \; \forall i \in I \; . \; \exists \hat{i} \in \hat{q}(\hat{s}) \; . \; f(s, i) \in \gamma(\hat{f}(\hat{s}, \hat{i})), \label{eq:soundness_1_sub_fine_to_coarse}\\
    &\forall \hat{s} \in \hat{S} \; . \; \forall s \in \gamma(\hat{s}) \; . \; \forall \hat{i} \in \hat{q}(\hat{s}) \; . \; \exists i \in I \; . \; f(s, i) \in \gamma(\hat{f}(\hat{s}, \hat{i})). \label{eq:soundness_1_sub_coarse_to_fine}
\end{align}
\end{subequations}

Informally, from each concrete state covered by an abstract state, \eqref{eq:soundness_1_sub_fine_to_coarse} requires that each concrete step result is covered by some abstract step result, and \eqref{eq:soundness_1_sub_coarse_to_fine} requires that each abstract step result covers some concrete step result.

 To prove \eqref{eq:soundness_1_sub_fine_to_coarse}, we assume $\hat{s} \in \hat{S}, s \in \gamma(\hat{s}), i \in I$ to be arbitrary but fixed. From \eqref{eq:soundness_sub_input_coverage}, we know that we can choose some $\hat{i} \in \hat{q}(\hat{s})$ for which $i \in \zeta(\hat{i})$.

To prove \eqref{eq:soundness_1_sub_coarse_to_fine}, we assume $\hat{s} \in \hat{S}, s \in \gamma(\hat{s}), \hat{i} \in \hat{q}(\hat{s})$ to be arbitrary but fixed. As $\zeta: \hat{I} \rightarrow 2^I \setminus \{\emptyset\}$, we can always choose some $i \in \zeta(\hat{i})$. 

In both situations, our assumptions include $s \in S, \hat{s} \in \gamma(\hat{s}), \hat{i} \in \hat{q}(\hat{s}), i \in \zeta(\hat{i})$, and it remains to prove $f(s, i) \in \gamma(\hat{f}(\hat{s}, \hat{i}))$. This follows from \eqref{eq:soundness_sub_step_coverage}.\qed
\end{proof}

%% file: sources/c_c_monotonicity.tex
\subsection{Proof of Monotonicity}
\label{app_subsec:monotonicity}

\begin{proof}[Theorem~\ref{th:monotonicity}]
We assume that $\hat{G}'=(\hat{S}, \hat{s}_0, \hat{I}, \hat{q}', \hat{f}', \hat{L})$ is a $(\gamma,\zeta)$-monotone refinement of the generating automaton $\hat{G}=(\hat{S}, \hat{s}_0, \hat{I}, \hat{q}, \hat{f}, \hat{L})$. We prove that for every \textmu-calculus property $\psi$ for which $\llbracket \psi \rrbracket(\MakePKS{\hat{G}}) \neq \bot$, also $\llbracket \psi \rrbracket(\MakePKS{\hat{G}'}) \neq \bot$. For this, it suffices to prove that $\MakePKS{\hat{G}}$ is sound wrt. $\MakePKS{\hat{G}'}$. Using
\begin{equation}
    H \defeq \{ (\hat{s}\rnew, \hat{s}\rold) \in \hat{S} \times \hat{S} \; | \; \gamma(\hat{s}\rnew) \subseteq \gamma(\hat{s}\rold) \},
\end{equation}
we will prove that \eqref{eq:lemma} holds for $G^\fine \defeq \hat{G}\rnew$, $G^\coarse \defeq \hat{G}\rold$. Formula~\eqref{eq:lemma_sub_initial} holds trivially, and~\eqref{eq:lemma_sub_labelling} holds due to \eqref{eq:labelling_monotonicity}. Assuming that $\hat{s}\rnew \in \hat{S}, \hat{s}\rold \in \hat{S}$ are arbitrary but fixed and $\gamma(\hat{s}\rnew) \subseteq \gamma(\hat{s}\rold)$ holds, we rewrite \eqref{eq:lemma_sub_fine_to_coarse} and \eqref{eq:lemma_sub_coarse_to_fine} to
\begin{subequations}
\label{eq:monotonicity_1}
\begin{align}
    & \forall \hat{i}\rnew \in \hat{q}\rnew(\hat{s}\rnew) \; . \; \exists \hat{i}\rold \in \hat{q}\rold(\hat{s}\rold) \; . \; \gamma(\hat{f}\rnew(\hat{s}\rnew, \hat{i}\rnew)) \subseteq \gamma(\hat{f}\rold(\hat{s}\rold, \hat{i}\rold)),\label{eq:monotonicity_a}\\
    &\forall \hat{i}\rold \in \hat{q}\rold(\hat{s}\rold) \; . \; \exists \hat{i}\rnew \in \hat{q}\rnew(\hat{s}\rnew) \; . \; \gamma(\hat{f}\rnew(\hat{s}\rnew, \hat{i}\rnew)) \subseteq \gamma(\hat{f}\rold(\hat{s}\rold, \hat{i}\rold)).\label{eq:monotonicity_b}
\end{align}
\end{subequations}

To prove~\eqref{eq:monotonicity_a}, we assume $\hat{i}\rnew \in \hat{q}\rnew(\hat{s}\rnew)$ arbitrary but fixed and choose $\hat{i}\rold \in \hat{q}\rold(\hat{s}\rold)$ for which $\zeta(\hat{i}\rnew) \subseteq \zeta(\hat{i}\rold)$, which exists due to \eqref{eq:monotonicity_sub_new_inputs}.

To prove~\eqref{eq:monotonicity_a}, we assume $\hat{i}\rold \in \hat{q}\rold(\hat{s}\rold)$ arbitrary but fixed and choose $\hat{i}\rnew \in \hat{q}\rnew(\hat{s}\rnew)$ for which $\zeta(\hat{i}\rnew) \subseteq \zeta(\hat{i}\rold)$, which exists due to \eqref{eq:monotonicity_sub_old_inputs}.

In both situations, it remains to prove $\gamma(\hat{f}\rnew(\hat{s}\rnew, \hat{i}\rnew)) \subseteq \gamma(\hat{f}\rold(\hat{s}\rold, \hat{i}\rold))$. From the assumed \eqref{eq:monotonicity_sub_step_coverage}, we strengthen this to $\gamma(\hat{f}\rold(\hat{s}\rnew, \hat{i}\rnew)) \subseteq \gamma(\hat{f}\rold(\hat{s}\rold, \hat{i}\rold))$. This follows from \eqref{eq:step_monotonicity}, given our other assumptions about $\hat{s}\rold, \hat{s}\rnew, \hat{i}\rold, \hat{i}\rnew$.\qed

\end{proof}

%% file: sources/c_d_completeness.tex
\subsection{Proof of Completeness}
\label{app_subsec:completeness}

\begin{proof}[Theorem~\ref{th:completeness}] We assume a sequence satisfying the preconditions of the theorem and show that it cannot be infinite. For each $\hat{G}$ in the sequence, we define a function $M_{\hat{G}}(\hat{s}) = (Q_{\hat{G}}(\hat{s}), F_{\hat{G}}(\hat{s}))$ where
\begin{equation}
\begin{split}
&Q_{\hat{G}}(\hat{s}) = \{ \zeta(\hat{i}) \; | \;  \hat{i} \in \hat{q}_{\hat{G}}(\hat{s}) \}, \; F_{\hat{G}}(\hat{s}) = \{ (\hat{i} \in \hat{I}, \gamma(\hat{f}_{\hat{G}}(\hat{s}, \hat{i}))) \},
\end{split}
\end{equation}
and will prove it cannot be equal for two different elements of the sequence named $\hat{G}$ and $\hat{G^*}$. We name the element after $\hat{G}$ as $\hat{G}'$, it is a strict ($\gamma$, $\zeta$)-monotone refinement of $\hat{G}$. $\hat{G^*}$ is a ($\gamma$, $\zeta$)-monotone refinement of both $\hat{G}$ and $\hat{G}'$ as Definition \ref{def:monotone} is clearly transitive and reflexive.

\textbf{Case 1.} The strict refinement is due to~\eqref{eq:input_tightening}. For some $\hat{s} \in \hat{S}$, there is an $\hat{i} \in \hat{q}(\hat{s})$ not covered by any $\hat{i}' \in \hat{q}(\hat{s})$, so $\zeta(\hat{i}) \in Q_{\hat{G}}(\hat{s})$ but $\zeta(\hat{i}) \not\in Q_{\hat{G}'}(\hat{s})$. Due to \eqref{eq:monotonicity_sub_new_inputs} from $\hat{G}'$ to $\hat{G}^*$, $\zeta(\hat{i}) \not\in Q_{\hat{G}^*}(\hat{s})$, so $M_{\hat{G}} \neq M_{\hat{G}^*}$.

\textbf{Case 2.} The strict refinement is due to~\eqref{eq:step_tightening}. There is some pair $(\hat{s}, \hat{i}) \in \hat{S} \times \hat{I}$ where for some $s \in \gamma(\hat{f}(\hat{s}, \hat{i}))$, $s \not\in \gamma(\hat{f}'(\hat{s}, \hat{i}))$. As such, $s \in F_{\hat{G}}(\hat{s})(\hat{i})$ but $s \not\in F_{\hat{G}'}(\hat{s})(\hat{i})$. Due to \eqref{eq:monotonicity_sub_step_coverage} from $\hat{G}'$ to $\hat{G}^*$, $s \not\in F_{\hat{G}^*}(\hat{s})(\hat{i})$, so $M_{\hat{G}} \neq M_{\hat{G}^*}$.

For each $\hat{s} \in \hat{S}$, $\hat{Q}(\hat{s})$ has less than $2^{|\hat{I}|}$ valuations. For each pair $(\hat{s}, \hat{i}) \in \hat{S} \times \hat{I}$, $F(\hat{s})(\hat{i})$ has less than $2^{|\hat{S}|}$ valuations. As such, there are less than $|S| (2^{|\hat{I}|} + |I| 2^{|\hat{S}|})$ valuations of $M$. Since we assume $\hat{S}, \hat{I}$ finite, this completes the proof.\qed

\end{proof}

%% file: sources/c_e_lemma.tex
\subsection{Proof of Lemma \ref{lemma:strict}}

\begin{proof}[Lemma \ref{lemma:strict}]

\textbf{First part.}
We will first prove the claim that if $\hat{G}$ is $(\gamma,\zeta)$-terminating wrt. $G$, then for every \textmu-calculus property $\psi$, $\llbracket \psi \rrbracket(\generatepks(\hat{G})) = \llbracket \psi \rrbracket(\generatepks(G))$. Since $G$ is a concrete generating automaton, $\llbracket \psi \rrbracket(\generatepks(G)) \neq \bot$, and it suffices to prove $\generatepks(G)$ is sound wrt. $\generatepks(\hat{G})$. Using Lemma \ref{lemma:lemma}, we define $H$ as
\begin{equation}
    H \defeq \{(\hat{s}, s) \; | \; \gamma(\hat{s}) = \{ s \}\}.
\end{equation}

\begingroup
\thinmuskip=0.6\thinmuskip
\medmuskip=0.6\medmuskip
\thickmuskip=0.6\thickmuskip
\eqref{eq:lemma_sub_initial} and \eqref{eq:lemma_sub_labelling} hold due to \eqref{eq:soundness_sub_initial} and \eqref{eq:completeness_sub_labelling}, respectively. We rewrite \eqref{eq:lemma_sub_fine_to_coarse} and \eqref{eq:lemma_sub_coarse_to_fine} as
\begin{subequations}
\label{eq:terminating_modal}
\begin{align}
    &\forall (\hat{s}, s) \in \hat{S} \times S \; . \; (\gamma(\hat{s}) = \{ s \} \Rightarrow \forall \hat{i} \in \hat{q}(\hat{s}) \; . \; \exists i \in I \; . \; \gamma(\hat{f}(\hat{s}, \hat{i})) = \{ f(s, i) \}), \label{eq:terminating_modal_sub_abstract_inputs}\\
    &\forall (\hat{s}, s) \in \hat{S} \times S \; . \; (\gamma(\hat{s}) = \{ s \} \Rightarrow \forall i \in I \; . \; \exists \hat{i} \in \hat{q}(\hat{s}) \; . \; \gamma(\hat{f}(\hat{s}, \hat{i})) = \{ f(s, i) \}). \label{eq:terminating_modal_sub_concrete_inputs}
\end{align}
\end{subequations}
\endgroup
We assume $(\hat{s}, s) \in \hat{S} \times S$ arbitrary but fixed and $\gamma(\hat{s}) = \{ s \}$ to hold.

To prove \eqref{eq:terminating_modal_sub_abstract_inputs}, we assume $\hat{i} \in \hat{q}(\hat{s})$ arbitrary but fixed and choose $i \in I$ for which $\zeta(\hat{i}) = \{ i \}$, which exists due to \eqref{eq:completeness_sub_abstract_inputs}.

To prove \eqref{eq:terminating_modal_sub_concrete_inputs}, we assume $i \in I$ arbitrary but fixed and choose $\hat{i} \in \hat{q}(\hat{s})$ for which $\zeta(\hat{i}) = \{ i \}$, which exists due to \eqref{eq:completeness_sub_concrete_inputs}.

\textbf{Second part.} We turn to the claim that if $\hat{G}$ is a soundness-guaranteeing $(\gamma,\zeta)$-abstraction of $G$ for which some $(\gamma,\zeta)$-monotone refinement $\hat{G}^*$ is $(\gamma,\zeta)$-terminating wrt. $G$, then either $\hat{G}$ itself is $(\gamma,\zeta)$-terminating wrt. $G$ or the refinement is strict. We will prove this by assuming $\hat{G}$ is not $(\gamma,\zeta)$-terminating and proving that $\hat{G}^*$ is a strict $(\gamma,\zeta)$-monotone refinement of $\hat{G}$.

As per Definition \ref{def:strict_monotone}, we are to prove either \eqref{eq:input_tightening} or \eqref{eq:step_tightening} holds for the refinement from $\hat{G}$ to $\hat{G}^*$. Since $\hat{G}$ is not $(\gamma,\zeta)$-terminating, at least one of \eqref{eq:completeness_sub_labelling},  \eqref{eq:completeness_sub_abstract_inputs}, \eqref{eq:completeness_sub_concrete_inputs}, \eqref{eq:completeness_sub_step} does not hold.

\textbf{Case (a).} \eqref{eq:completeness_sub_labelling} does not hold for $\hat{G}$. However, $\hat{S}$, $\hat{I}$, and $\hat{L}$ are the same for $\hat{G}^*$, and we already assumed \eqref{eq:completeness_sub_labelling} for $\hat{G}^*$. This is a contradiction and so this case cannot occur.

\textbf{Case (b).} \eqref{eq:completeness_sub_abstract_inputs} does not hold for $\hat{G}$, but we have assumed it holds for $\hat{G}^*$, i.e.
\begin{subequations}
\begin{align}
    &\exists \hat{s} \in \hat{S} \; . \; \exists \hat{i} \in \hat{q}(\hat{s}) \; . \; \forall i \in I \; . \; \zeta(\hat{i}) \neq \{i\}. \label{eq:case_b_sub_nonterminating}\\
    &\forall \hat{s} \in \hat{S} \; . \; \forall \hat{i}^* \in \hat{q}^*(\hat{s}) \; . \; \exists i^* \in I \; . \; \zeta(\hat{i}^*) = \{i^*\}. \label{eq:case_b_sub_terminating}
\end{align}
\end{subequations}
We will prove that \eqref{eq:input_tightening} holds for the refinement from $\hat{G}$ to $\hat{G}^*$, i.e.
\begin{equation}
\begin{split}
    &\exists \hat{s} \in \hat{S} \; . \; \exists \hat{i} \in \hat{q}(\hat{s}) 
\; . \; \forall \hat{i}' \in \hat{q}^*(\hat{s}) \; . \; \exists i \in \zeta(\hat{i}) \; . \; i \not\in \zeta(\hat{i}'). \label{eq:input_tightening_again}
\end{split}
\end{equation}
We fix $\hat{s}$, $\hat{i}$ from \eqref{eq:case_b_sub_nonterminating} and choose them in \eqref{eq:input_tightening_again}. For all $\hat{i}' \in \hat{q}^*(\hat{s})$, clearly $|\gamma(\hat{i}')| = 1$ due to \eqref{eq:case_b_sub_terminating}. However, due to \eqref{eq:case_b_sub_nonterminating} and $\zeta$ not returning $\emptyset$, $|\gamma(\hat{i})| > 1$. As such, for any $\hat{i}' \in \hat{q}^*(\hat{s})$, there exists $i \in \zeta(\hat{i})$ such that $i \not\in \zeta(\hat{i}')$.

\textbf{Case (c).} \eqref{eq:completeness_sub_concrete_inputs} does not hold for $\hat{G}$, i.e.
\begin{equation}
    \exists \hat{s} \in \hat{S} \; . \; \exists i \in I \; . \; \forall \hat{i} \in \hat{q}(\hat{s})  \; . \; \zeta(\hat{i}) \neq \{i\}. \label{eq:case_c}
\end{equation}
Since $\hat{G}$ is assumed to be a soundness-guaranteeing $(\gamma,\zeta)$-abstraction of $G$, it must fulfil the full input coverage requirement from \eqref{eq:soundness_sub_input_coverage}, i.e.
\begin{equation}
        \forall (\hat{s}, i) \in \hat{S} \times I \; . \; \exists \hat{i} \in \hat{q}(\hat{s}) \; . \; i \in \zeta(\hat{i}).\label{eq:case_c_requirement}
\end{equation}
Fixing $\hat{s} \in \hat{S}, i \in I$ from \eqref{eq:case_c} and choosing them in \eqref{eq:case_c_requirement}, we combine both formulas to obtain
\begin{equation}
    (\forall \hat{i} \in \hat{q}(\hat{s})  \; . \; \zeta(\hat{i}) \neq \{i\}) \land (\exists \hat{i} \in \hat{q}(\hat{s}) \; . \; i \in \zeta(\hat{i})).
\end{equation}
We can weaken this to 
\begin{equation}
    \exists \hat{i} \in \hat{q}(\hat{s}) \; . \; (i \in \zeta(\hat{i}) \land \zeta(\hat{i}) \neq \{i\}).
\end{equation}
This implies $\exists \hat{i} \in \hat{q}(\hat{s}) \; . \; |\gamma(\hat{i})|>1$, where $i$ no longer appears. We reintroduce an existential quantifier for $\hat{s}$,
\begin{equation}
    \exists \hat{s} \in \hat{S} \; . \; \exists \hat{i} \in \hat{q}(\hat{s}) \; . \; |\zeta(\hat{i})|>1.
\end{equation}
This can be weakened to
\begin{equation}
    \exists \hat{s} \in \hat{S} \; . \; \exists \hat{i} \in \hat{q}(\hat{s}) \; . \; \forall i \in I \; . \; \zeta(\hat{i}) \neq \{ i \},
\end{equation}
which is the same as \eqref{eq:case_b_sub_nonterminating} and we can complete proving this case using the argument from Case (b).

\textbf{Case (d).} \eqref{eq:completeness_sub_step} does not hold for $\hat{G}$ but holds for $\hat{G}^*$, i.e. it holds that
\begin{subequations}
\begin{align}
    &\exists (\hat{s}, \hat{i}, s, i) \in \hat{S} \times \hat{I} \times S \times I \; . \; ((\gamma(\hat{s}), \zeta(\hat{i})) = (\{ s \}, \{i\}) \land \gamma(\hat{f}(\hat{s}, \hat{i})) \neq \{ f(s,i) \}), \label{eq:case_d_sub_nonterminating}\\
    &\forall (\hat{s}, \hat{i}, s, i) \in \hat{S} \times \hat{I} \times S \times I \; . \; ((\gamma(\hat{s}), \zeta(\hat{i})) = (\{ s \}, \{i\}) \Rightarrow \gamma(\hat{f}^*(\hat{s}, \hat{i})) = \{ f(s,i) \}). \label{eq:case_d_sub_terminating}
\end{align}
\end{subequations}
We fix $\hat{s}, \hat{i}, s, i$ from \eqref{eq:case_d_sub_nonterminating} and choose them in \eqref{eq:case_d_sub_terminating}, obtaining
\begin{subequations}
\begin{align}
    &(\gamma(\hat{s}), \zeta(\hat{i})) = (\{ s \}, \{i\}) \land \gamma(\hat{f}(\hat{s}, \hat{i})) \neq \{ f(s,i) \}, \label{eq:case_d_1_sub_nonterminating}\\
    &(\gamma(\hat{s}), \zeta(\hat{i})) = (\{ s \}, \{i\}) \Rightarrow \gamma(\hat{f}^*(\hat{s}, \hat{i})) = \{ f(s,i) \}. \label{eq:case_d_1_sub_terminating}
\end{align}
\end{subequations}
We therefore know $\gamma(\hat{f}(\hat{s}, \hat{i})) \neq \gamma(\hat{f}^*(\hat{s}, \hat{i}))$ and $|\gamma(\hat{f}^*(\hat{s}, \hat{i}))| = 1$. Due to assumed monotonicity, we also know from \eqref{eq:monotonicity_sub_step_coverage} that $\gamma(\hat{f}^*(\hat{s}, \hat{i})) \subseteq \gamma(\hat{f}(\hat{s}, \hat{i}))$. Therefore, it must be the case that
\begin{subequations}
\begin{align}
    \exists s \in \gamma(\hat{f}^*(\hat{s}, \hat{i})) \; . \; s \not\in \gamma(\hat{f}^*(\hat{s}, \hat{i})).
\end{align}
\end{subequations}
Reintroducing the quantifiers for $\hat{s}, \hat{i}$, we obtain \eqref{eq:step_tightening} from $\hat{G}$ to $\hat{G}^*$, which completes the proof. \qed
\end{proof}